\documentclass[a4paper,USenglish,cleveref, autoref, thm-restate]{lipics-v2021}
\usepackage{graphicx}
\usepackage{amsmath}
\usepackage{amsfonts}
\usepackage{amssymb}
\usepackage{graphicx}
\usepackage{tikz}
\usepackage{amsthm}
\usepackage{stmaryrd}
\usepackage{float}
\usepackage{caption}
\usepackage{subcaption}
\usepackage{enumerate}
\usepackage{autobreak}
\usepackage{cases}
\usepackage{placeins}
\usepackage{algorithm}
\usepackage{algpseudocode}
\usepackage{amsmath}
\usepackage{todonotes}
\usepackage{layouts}
\usepackage{nicefrac}
\usepackage{xurl}
\usepackage[binary-units=true]{siunitx}
\usepackage[capitalise, noabbrev]{cleveref}
\crefname{algocf}{Algorithm}{Algorithms}
\Crefname{algocf}{Algorithm}{Algorithms}
\crefname{@theorem}{Theorem}{Theorems}
\Crefname{@theorem}{Theorem}{Theorems}

\newcommand{\atsp}{\alpha_{\mbox{\scriptsize{TSP}}}}
\newcommand{\old}[1]{{}}

\title{The Lawn Mowing Problem: From Algebra to Algorithms} %TODO Please add

\titlerunning{The Lawn Mowing Problem: From Algebra to Algorithms}
\author{Sándor P. Fekete}{Department of Computer Science, TU Braunschweig, Germany}{s.fekete@tu-bs.de}{https://orcid.org/0000-0002-9062-4241}{}
\author{Dominik Krupke}{Department of Computer Science, TU Braunschweig, Germany}{d.krupke@tu-bs.de}{https://orcid.org/0000-0003-1573-3496}{}
\author{Michael Perk}{Department of Computer Science, TU Braunschweig, Germany}{perk@ibr.cs.tu-bs.de}{https://orcid.org/0000-0002-0141-8594}{}
\author{Christian Rieck}{Department of Computer Science, TU Braunschweig, Germany}{rieck@ibr.cs.tu-bs.de}{https://orcid.org/0000-0003-0846-5163}{}
\author{Christian Scheffer}{Faculty of Electrical Engineering and Computer Science, Bochum University of Applied Sciences, Bochum, Germany}{christian.scheffer@hs-bochum.de}{https://orcid.org/0000-0002-3471-2706}{}
\authorrunning{S. P. Fekete, D. Krupke, M.  Perk, C. Rieck, C. Scheffer}

 %TODO mandatory: Please choose ACM 2012 classifications from https://dl.acm.org/ccs/ccs_flat.cfm

\Copyright{The authors} %TODO mandatory, please use full first names. LIPIcs license is "CC-BY";  http://creativecommons.org/licenses/by/3.0/

 %TODO mandatory; please add comma-separated list of keywords

\category{} %optional, e.g. invited paper

\supplement{\url{https://github.com/tubs-alg/lawn-mowing-from-algebra-to-algorithms}}%optional, e.g. related research data, source code, ... hosted on a repository like zenodo, figshare, GitHub, ...
\funding{Work at TU Braunschweig has been partially supported by the German Research Foundation (DFG),
project ``Computational Geometry: Solving Hard Optimization Problems'' (CG:SHOP), grant FE407/21-1.}

\acknowledgements{}%optional

\nolinenumbers %uncomment to disable line numbering
\hideLIPIcs

\newcommand{\Pol}{P}
\newcommand{\T}{T}

\newcommand{\R}{\mathbb{R}}

\newcommand{\Q}{\mathbb{Q}}

\definecolor{dkgreen}{rgb}{0,0.6,0}
\definecolor{gray}{rgb}{0.5,0.5,0.5}
\definecolor{mauve}{rgb}{0.58,0,0.82}

\definecolor{functioncolor}{rgb}{0,0,0}
\definecolor{keywordcolor}{rgb}{0.58,0,0.82}
\definecolor{variablecolor}{rgb}{0,0.1,.8}
\lstdefinelanguage{Magma}{
	keywords={},
	literate={
		{R}{{{\color{variablecolor}{R}}}}{1}
		{f}{{{\color{variablecolor}{f}}}}{1}
		{G}{{{\color{variablecolor}{G}}}}{1}
		{GroupName}{{{\color{keywordcolor}{GroupName}}}}{1}
		{RationalField}{{{\color{keywordcolor}{RationalField}}}}{1}
		{PolynomialRing}{{{\color{keywordcolor}{PolynomialRing}}}}{1}
		{GaloisGroup}{{{\color{keywordcolor}{GaloisGroup}}}}{1}
		{IsSolvable}{{{\color{keywordcolor}{IsSolvable}}}}{1}
	},
	numbers=left,
	numberstyle=\tiny\color{gray},
	numbersep=5pt,
	breaklines=true,
	captionpos={t},
	frame={lines},
	rulecolor=\color{black},
	framerule=0.5pt,
	columns=flexible,
	tabsize=2,
}

\begin{document}

\maketitle

\begin{abstract}
For a given polygonal region $\Pol$, the Lawn Mowing Problem (LMP) asks for a
shortest tour $\T$ that gets within Euclidean distance 1/2 of every point in
$\Pol$; this is equivalent to computing a shortest tour for a unit-diameter cutter
$C$ that covers all of $\Pol$. As a generalization of the Traveling Salesman Problem,
the LMP is NP-hard; unlike the discrete TSP, however, the LMP has defied
efforts to achieve exact solutions, due to its combination of combinatorial complexity
with continuous geometry. 
%making even verification of a given solution a non-trivial issue.

We provide a number of new contributions that provide insights into 
the involved difficulties, as well as positive results that enable
both theoretical and practical progress.
(1) We show that the LMP is algebraically hard: it is not solvable by radicals
over the field of rationals, even for the simple case in which $\Pol$ is
a $2\times 2$ square. This implies that it is impossible to compute
exact optimal solutions under models of computation
that rely on elementary arithmetic operations
and the extraction of $k$th roots, and explains the perceived practical difficulty.
(2) We exploit this algebraic analysis for the natural
class of polygons with axis-parallel edges and integer vertices (i.e.,
polyominoes), highlighting the relevance of turn-cost minimization
for Lawn Mowing tours, and leading to a general construction method
for feasible tours.
(3) We show that this construction method achieves
theoretical worst-case guarantees that improve
previous approximation factors for polyominoes.
(4)~We demonstrate the practical usefulness \emph{beyond polyominoes} 
by performing an extensive practical study on a spectrum
of more general benchmark polygons: We obtain solutions that are %significantly
better than the previous best values by Fekete et al., for instance sizes up to
20 times larger.  
\end{abstract}

\section{Introduction}
Many geometric optimization problems are NP-hard:
the number of possible solutions is finite, but there may
not be an efficient method for systematically finding a best one.
A~different kind of difficulty considered in geometry
%long before the invention of computers, 
is rooted 
%in the asymptotic complexity of finding the best of a finite
%number of candidates, but 
in the impossibility of obtaining solutions with a given set of construction tools:
Computing the length of a diagonal of a square is not possible
with only rational numbers; trisecting any given angle cannot be done
with ruler and compass, and neither can a square be constructed whose area
is equal to that of a given circle.

In this paper, we consider
the \emph{Lawn Mowing Problem} (LMP), in which we are
given a %(not necessarily simple or even connected) 
polygonal region $\Pol$ and
a disk cutter $C$ of diameter~$1$; the task is to find 
a closed roundtrip
%(a \emph{tour}) 
of minimum Euclidean length, such that the cutter ``mows'' all of $\Pol$
, i.e.,
a~shortest tour that moves the center of $C$ within distance $1/2$ from every point in~$\Pol$. 
The LMP naturally occurs in a wide spectrum of practical applications, such as
robotics, manufacturing, farming, quality control, and image processing, so it
is of both theoretical and practical importance. As a
generalization of the classic Traveling Salesman Problem (TSP), the LMP is also
NP-hard; however, while the TSP has shown to be amenable to exact methods for
computing provably optimal solutions even for large instances~\cite{applegate2007traveling}, the LMP has
defied such attempts, with only recently some first practical progress
by Fekete et al.~\cite{fekete2022closer}.

\subsection{Related Work}
There is a wide range of practical applications for the LMP, including
manufacturing~\cite{Arkin2000ZigZag,Held1991,Held},
cleaning~\cite{bormann2015new}, robotic
coverage~\cite{cabreira2019survey,choset2001coverage,galceran2013survey,jensen2020near},
inspection~\cite{englot2012sampling}, CAD~\cite{elber1999offsets},
farming~\cite{bahnemannrevisiting,choset1998coverage,oksanen2009coverage}, and
pest control~\cite{zika}.  In Computational Geometry, the Lawn Mowing Problem
was first introduced by Arkin et al.~\cite{arkin1993lawnmower}, who later gave
the currently best approximation algorithm with a performance guarantee of
$2\sqrt{3}\atsp\approx 3.46 \atsp$~\cite{Arkin2000}, where $\atsp$ is the
performance guarantee for an approximation algorithm for the TSP.

Optimally covering even relatively simple regions such as a disk by a set of $n$ \emph{stationary} unit disks has received
considerable attention, but is excruciatingly difficult; see~\cite{bezdek1979korok,bezdek1984einige,friedman1014,heppes1997covering,melissen2000covering,neville1915solution}.
%%gave optimal solutions for~$n\leq 5$;
%Melissen and Schuur~\cite{melissen2000covering}, extended this
%for $n=6,7$.  See Friedman~\cite{friedman1014} for %illustrations of
%the best known solutions (only some of which are proven to be optimal) for $n\leq 12$.
%As early as 1915, Neville~\cite{neville1915solution} computed the optimal arrangement for covering a disk by five unit disks,
%but reported a wrong optimal value; much later, Bezdek~\cite{bezdek1979korok,bezdek1984einige} gave the correct value for $n=5,6$.
As~recently as 2005, Fejes T\'{o}th~\cite{toth2005thinnest} established optimal values for the maximum radius of a disk that can be covered by $n=8,9,10$ unit circles.
%The question of incomplete coverings was raised
%in 2008 by Connelly;
%who asked how one should place $n$ small disks of radius
%$r$ to cover the largest possible area of a disk of radius $R > r$.
%Szalkai~\cite{szalkai2016optimal} gave an optimal solution for $n=3$.
Recent progress on covering by (not necessarily equal) disks has been achieved by Fekete et al.~\cite{2020-Covering_SoCG,75-Covervideo_SoCG}.

A first \emph{practical} breakthrough on computing provably good Lawn Mowing
tours was achieved by Fekete et al.~\cite{fekete2022closer},
who established a primal-dual algorithm for the LMP by iteratively covering an
expanding \emph{witness set} of finitely many points in $\Pol$. In each iteration,
their method computes a lower bound, which involves solving a special case of a TSP instance with neighborhoods,
the \emph{Close-Enough TSP} (CETSP) to provable optimality; for an upper bound, the method
is enhanced to provide full coverage. In each iteration, this establishes 
both a valid solution and a valid lower bound, and thereby a bound on the remaining optimality gap.
They also provided a computational study,
with good solutions for a large spectrum of benchmark instances with up to
\num{2000} vertices. However, this approach encounters scalability issues for
larger instances, due to the considerable number of witnesses that need to 
be placed.

A seminal result on algebraic aspects of geometric optimization problems was achieved by
Bajaj~\cite{bajaj1988algebraic}, who established algebraic hardness
for the Fermat-Weber problem of finding a point in
$\mathbb{R}^2$ that minimizes the sum of Euclidean distances to
all points in a given set. 
Others have studied the Galois complexity for geometric problems like Graph Drawing or the Weighted Shortest Path Problem~\cite{DBLP:journals/jgaa/BannisterDEG15,DBLP:journals/comgeo/CarufelGMOS14,DBLP:conf/cvpr/NisterHS07}.
%Note, however, that the Fermat-Weber problem is relatively benign in
%practical difficulty, as it amounts to minimizing a smooth, convex function
%over a compact set, which can be achieved with high accuracy
%by using a numerical approach such as Newton's method. This was exploited
%for algorithmic purposes by Fekete et al.~\cite{fmr+-shpae-02}.
%The use of straight-edge and compass is known
%to be equivalent to the use of $(+,-,*,/,\sqrt{})$ over
%$\Q$~\cite{courant1941mathematics}.
%Our first main result implies that the Lawn Mowing Problem is not solvable by radicals
%over $\Q$, i.e., a solution is not expressible in terms of
%$(+,-,*,/,\sqrt[k]{})$ over $\Q$.

%In the course of the algorithmic analysis, we will encounter the problem
%of covering tours with turn cost; 
As we will see in the course of our algorithmic analysis the number of turns in
a tour is of crucial importance for the overall cost; this has been previously
studied by Arkin et al.~\cite{arkin2005optimal} in a discrete setting. 
This objective is also of practical importance in the context of physical coverage, e.g.,
in the context of efficient drone trajectories~\cite{zika}.

\subsection{Our Results}
We provide a spectrum of new theoretical and practical results
 %that give deeper insights into
for the Lawn Mowing Problem.

\begin{itemize}
\item
%We establish a fundamental reason for the perceived geometric
%difficulty of characterizing optimal solutions for
%the LMP: 
We prove that computing an optimal Lawn Mowing tour
is algebraically hard, even for the 
case of mowing a $2\times 2$ square by a unit-diameter disk, as it
requires computing zeroes of high-order irreducible polynomials.
% making %. As a consequence,
%computing even near-optimal solutions for the LMP
%requires dealing with algebraic issues of numerical approximation and
%accuracy, making even small instances of the LMP fundamentally 
%the LMP more challenging than its special case, the discrete TSP\@.
\item
We exploit the algebraic analysis
% exploit the previous algebraic analysis of the structure of good solutions
to achieve provably good trajectories %for the natural class of polygons with axis-parallel edges
%and integer vertices (i.e., 
for polyominoes, based on the consideration of turn cost, and provide a method for general polygons.
%This
%highlights the relevance of turn-cost minimization
%for Lawn Mower tours, and leads to a general construction method
%for feasible tours.
\item
We show that this construction method achieves
theoretical worst-case guarantees that improve
previous approximation factors for polyominoes.
\item 
We demonstrate the practical usefulness \emph{beyond polyominoes} 
%by performing an extensive practical study 
on a spectrum of more general benchmark polygons, obtaining better solutions 
than the previous values by Fekete et al.~\cite{fekete2022closer}, for instance sizes up to
20 times larger. 
\end{itemize}

\subsection{Definitions}
A~(simple)~\emph{polygon}~$\Pol$ is a (non-self-intersecting) shape in the
plane, bounded by a finite number $n$ of line segments. The \emph{boundary} of
a polygon~$\Pol$ is denoted by $\partial \Pol$. A \emph{polyomino}
is a polygon with axis-parallel edges and vertices with integer coordinates;
any polyomino can be canonically partitioned into a finite number $N$ of
unit-squares, called \emph{pixels}. A~\emph{tour} is a closed
continuous curve $T: [0,1] \rightarrow \mathbb{R}^2$ with $T(0) = T(1)$.
%we denote the (Euclidean) length of $T$ by $\ell(T)$.
The~\emph{cutter}~$C$ is a disk of diameter $d$, centered in its midpoint. W.l.o.g., we assume $d=1$ for the rest of the paper.
%The
%\emph{Minkowski sum} of two sets $A, B\subset \mathbb{R}^2$ is the set $A\oplus
%B = \{a+b\ |\  a\in A,\,b\in B\}$. 
The \emph{coverage} of a tour $T$ with the disk cutter $C$ is the Minkowski
sum $T \oplus C$. %The coverage of a point $p \in T$ is $\{
%p \} \oplus C$. 
A \emph{Lawn Mowing tour}~$T$ of a polygon $P$ with a cutter $C$ is a tour 
whose coverage contains $P$. An \emph{optimal} Lawn Mowing tour
%$\Topt$ 
is a Lawn Mowing tour of shortest length. 
%For a discrete set of points
%$P$\todo{$P$ is a set of points and the polygon...}, a tour $T$ \emph{traverses} $P$ if $P \subset T$, i.e., for each \new{point $p$} %$p \in P$
%there is a $t \in [0,1]$ such that $T(t) = p$.

% We will define some commands that we will use throughout this file.
\newcommand{\optPath}{\omega}
\newcommand{\squarePol}{\Pol}
\newcommand{\sqnoindex}{S}
\newcommand{\sq}{\sqnoindex_0}
\newcommand{\enterPoint}{p_s}
\newcommand{\upperLeftPoint}{q}
\newcommand{\fociPoint}{p_\delta}
\newcommand{\leavePoint}{p_t}
\newcommand{\sqBottomRight}{s_1}
\newcommand{\sqBottomLeft}{s_0}
\newcommand{\sqTopRight}{s_2}
\newcommand{\sqTopLeft}{s_3}
\newcommand{\quadrantBottomLeft}{S_{0,0}}
\newcommand{\quadrantBottomRight}{S_{0,1}}
\newcommand{\quadrantTopRight}{S_{0,2}}
\newcommand{\quadrantTopLeft}{S_{0,3}}
\newcommand{\uncoveredSegment}{\varepsilon}
\newcommand{\ellipse}{E}
\newcommand{\ellipsePathLength}{c}
\newcommand{\ellipseCenter}{p_c}
\newcommand{\ellipseFociDistance}{d_E}
\newcommand{\ellipseMajor}{a}
\newcommand{\ellipseMinor}{b}
\newcommand{\optTourSquare}{T}
\newcommand{\lemmaPathCircle}{U}
\newcommand{\subsquare}{\lambda}
\newcommand{\subsquareVertex}{\lambda}
%
%
% Start of the section
\section{Algebraic Hardness}\label{sec:tours-in-rectangles}

In their recent work, 
Fekete~et~al.~\cite{fekete2022closer} prove that an optimal
Lawn Mowing tour for a polygonal region is necessarily polygonal
itself; on the other hand, they show that optimal tours may 
need to contain vertices with irrational coordinates corresponding to arbitrary square roots, 
even if $\Pol$ is just a triangle. In the following we show 
that if $\Pol$ is a $2\times 2$ square, an optimal tour may involve coordinates
that cannot even be described with radicals.
See \cref{fig:opt} for the structure of optimal trajectories.

\begin{restatable}{theorem}{notSolvableByRadicals}\label{theorem:not-solvable-by-radicals}
  For the case in which $\squarePol$ is a $2\times 2$ square, the Lawn Mowing Problem
is algebraically hard: an optimal tour involves coordinates that are zeroes
of polynomials that cannot be expressed by radicals.
\end{restatable}

A key 
observation is that covering each of the four corners 
$(0,-1), (2,-1), (2,1), (0,1)$ of a $2\times 2$ square $\squarePol$
requires the disk center to leave the subsquare $\subsquare$ with vertices
$\subsquareVertex_0=(1/2,-1/2)$, 
$\subsquareVertex_1=(3/2,-1/2)$, 
$\subsquareVertex_2=(3/2,1/2)$, 
$\subsquareVertex_3=(1/2,1/2)$, 
obtained by offsetting the boundary of $\squarePol$ by the radius of $C$,
which is the locus of all disk centers for which $\subsquare$ stays inside $\squarePol$.
However, covering the area close to the center of $\squarePol$ also requires keeping
the center of $C$ within $\subsquare$; as we argue in the following, this results in
a trajectory with a ``long''
portion (shown vertically in the figure) for which the disk covers the center of $\squarePol$ and the boundary
of $C$ traces the boundary of $\squarePol$, and a ``short'' portion for
which $C$ only dips into $\subsquare$ without tracing the  boundary of $\squarePol$.
%The proof can be separated into several sections.

%We start our proof by considering an optimal lawn mowing tour for a 
%rectangle and then argue why no solution can be obtained in 
%terms of $(+,-,*,/,\sqrt[k]{})$ over $\Q$.

\subsection{Optimal Tours at Corners}
\label{subsec:corner}
%To find an optimal tour within a rectangle, we subdivide it into smaller regions.
For the $2\times 2$ square $\squarePol$, consider the upper left $1\times 1$ subsquare
$\sq$ with corners $(0,0)$, $(0,1)$, $(1,1)$, $(0,1)$, further subdivided into
four 
%We will start by finding an optimal path through a 
$1/2\times 1/2$ quadrants $\quadrantBottomLeft,\ldots,\quadrantTopLeft$, as
shown in \cref{fig:2x2-square}, and an optimal path $\optPath$ that enters
$\sq$ at the bottom
and leaves it to the right.  Let $\enterPoint=(\enterPoint^x, 0)$, $\leavePoint=(1, \leavePoint^y)$ 
be the points where $\optPath$ enters and
leaves $\sq$, respectively.
For the following lemmas, we assume that a covering path exists that obeys the above
conditions. We will later determine that path and show that it covers $\sq$.

\begin{restatable}{lemma}{optPathCenterPoint}\label[lemma]{lemma:optPathCenterPoint}
$\enterPoint^x \leq 1/2$ and $\leavePoint^y\geq 1/2$ and either $\enterPoint^x=1/2$ or $\leavePoint^y=1/2$.
\end{restatable}

\begin{proof}
  To cover $\sqBottomRight$, $\optPath$ must intersect a circle with diameter $1$ 
  centered in $\sqBottomRight$.
  Any path with $\enterPoint$ right of $(1/2,0)$ or $\leavePoint$ below $(1,1/2)$ can
  be made shorter by shifting the point $\enterPoint$ to $(1/2,0)$ or
  $\leavePoint$ to $(1,1/2)$.  Any path with $\enterPoint$ left of $(1/2,0)$ and
  $\leavePoint$ above $(1,1/2)$ must enter $\quadrantBottomRight$, resulting in a detour.
\end{proof}

Without loss of generality, we assume that $\enterPoint^x=1/2$. The next step is to find the optimal position of $\leavePoint$. 
As an optimal path $\optPath$ must enter the quadrant $\quadrantTopLeft$ once, we can subdivide the path into two parts.
For some $\delta > 0$, let $\leavePoint^y = 1/2 + \delta$ and $\fociPoint=(1/2,\delta)$.

\begin{figure}[t]
  \begin{subfigure}[b]{.5\linewidth}
      \centering
      \includegraphics[width=.66\linewidth, page=1]{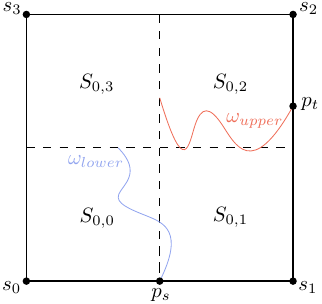}
      \caption{The upper left $1\times 1$ subsquare $\sq$ of $\squarePol$.}\label{fig:2x2-square}
  \end{subfigure}
  \begin{subfigure}[b]{.5\linewidth}
      \centering
      \includegraphics[width=.66\linewidth, page=2]{figures/2x2_square.pdf}
      \caption{Visualization of \cref{lemma:optPathDeltaSegment}.}\label{fig:2x2-square-case-i}
  \end{subfigure}
  \caption{Computing an optimal path $\optPath$ through the square $\sq$.}
  \label{fig:square-2-2}
  \end{figure}

\begin{restatable}{lemma}{optPathDeltaSegment}\label[lemma]{lemma:optPathDeltaSegment}
For any $\delta > 0$, $\optPath$ has a subpath $\enterPoint\fociPoint$.
\end{restatable}
\begin{proof}
  We denote the part from $\enterPoint$ to $\quadrantTopLeft$ as the \emph{lower portion} and 
  from $\quadrantTopLeft$ to $\leavePoint$ as the \emph{upper portion} of $\optPath$, see~\cref{fig:square-2-2}.
  Let $\sqBottomRight'=(1, \delta)$ and
  $\uncoveredSegment=\sqBottomRight\sqBottomRight'$.  Segment $\uncoveredSegment$
  must be covered by $\optPath$. We distinguish two cases; (i)
  $\uncoveredSegment$ is covered by the lower portion of $\optPath$ or (ii)
  $\uncoveredSegment$ is covered by the upper portion of $\optPath$.  For case
  (i), let us assume that $\uncoveredSegment$ is covered by the lower portion of
  $\optPath$.  When $\optPath$ would enter $\quadrantBottomRight$ it would also
  have to enter $\quadrantBottomLeft$ to cover the left side of
  $\quadrantBottomLeft$.  It is clear that traversing the segment
  $\enterPoint\fociPoint$ of length $\delta$ is the best way to cover the lower
  portion of $\quadrantBottomLeft, \quadrantBottomRight$, as any other path would
  need additional segments in $x$-direction, see~\cref{fig:2x2-square-case-i}.  Any
  path that obeys case (ii) is suboptimal, as it has to cover $\uncoveredSegment$
  from within $\quadrantTopRight$, for a detour of at least~$2\delta$.
\end{proof}
  
\begin{lemma} \label[lemma]{lemma:optimal-path-square-with-circle}
  The uniquely-shaped optimal Lawn Mowing path $\optPath$ between two adjacent sides of $\sq$ has length $L_{\sq}\approx 1.309$ with $\optPath=(\enterPoint, \fociPoint, \upperLeftPoint, \leavePoint)$ and
  \begin{equation*}
    \enterPoint = (\frac{1}{2}, 0) \quad%
    \fociPoint = (\frac{1}{2}, \delta) \approx (\frac{1}{2},0.168) \quad%
    \upperLeftPoint \approx (0.386, 0.682) \quad%
    \leavePoint = (1, \frac{1}{2}+\delta) \approx (1, 0.668).
    \end{equation*}
\end{lemma}
\begin{proof}
Let $\sqTopLeft$ be the top left corner of $\sq$.
We identify a shortest path for visiting one point $\upperLeftPoint$ on a circle 
$\lemmaPathCircle$ with diameter $1$
centered in $\sqTopLeft$ 
dependent on $\delta$, a necessary condition for a feasible path.
Let $\ellipsePathLength = d(\fociPoint, \upperLeftPoint) + d(\upperLeftPoint, \leavePoint)$ be the distance from both points to $\lemmaPathCircle$.
Consider an ellipse $\ellipse$ with foci $\fociPoint, \leavePoint$ that touches~$\lemmaPathCircle$ in a single point, see \cref{fig:ellipse-2x2-square}.
By definition, %of an ellipse, 
the intersection point $\upperLeftPoint$ minimizes the distance $\ellipsePathLength$. 
For $\delta \in [0,1]$ we want to find an intersection point between $\ellipse$ and $\lemmaPathCircle$ that minimizes distance $\ellipsePathLength$.
Let $\ellipseCenter=(\ellipseCenter^x,\ellipseCenter^y)$ be the center point of $\ellipse$ and $\ellipseFociDistance$ be the distance from 
the center point of $\ellipse$ and $\ellipseMajor,\ellipseMinor$ 
the major/minor axis.

\begin{equation}\label{eq:ellipse-center}
\ellipseCenter^x = \frac{3}{4} \quad%
\ellipseCenter^y = \frac{1}{4} + \delta \quad%
\ellipseFociDistance = d(\fociPoint, \ellipseCenter) = \frac{\sqrt{2}}{4} \quad%
\ellipseMajor = \frac{1}{2} \ellipseFociDistance \quad%
\ellipseMinor = \sqrt{\ellipseMajor^2 - \ellipseFociDistance^2}
\end{equation}

The ellipse can now be defined with its center point $\ellipseCenter$, the major/minor axis $\ellipseMajor,\ellipseMinor$ and 
the angle $\theta$, which is the angle between a line through $\fociPoint, \leavePoint$ and the $x$-axis.
We formulate the shortest path problem as a minimization problem while inserting \cref{eq:ellipse-center}.

\begin{align*}
  &\text{min}  && \ellipsePathLength + \delta &\\
  &\text{s.t.} && x^2+(y-2)^2 - \frac{1}{4} &~=~ 0\\[6pt]
  &&& \dfrac {((x-\ellipseCenter^x)\cos(\theta)+(y-\ellipseCenter^y)\sin(\theta))^2}{\ellipseMajor^2}+\dfrac{((x-\ellipseCenter^x) \sin(\theta)-(y-\ellipseCenter^y) \cos(\theta))^2}{\ellipseMinor^2}&~=~1\\
  &&& \sqrt{(x-1)^2 + (y-\delta)^2} + \sqrt{(x-2)^2 + (y-1-\delta)^2} - c &~=~ 0
  \end{align*}
  
  The objective minimizes the total length of the path $\optPath$ with variables that encode the exact coordinates of $\fociPoint, \upperLeftPoint, \leavePoint$.
  An intersection point of $\ellipse$ and $C$ with center $\sqTopLeft=(0,1)$ is a solution to the first and second constraints, respectively.
  An exact optimization approach using Mathematica reveals that $\delta, \upperLeftPoint^x,\upperLeftPoint^y$ can only be expressed as the first, third, 
  and first roots of three irreducible high-degree polynomials $f_\delta, f_{\upperLeftPoint^x}, f_{\upperLeftPoint^y}$, see \cref{eq:f-delta,eq:f-x,eq:f-y}.
  
  \begin{align}
    f_\delta(x)= & \num{589824}x^{16} - 
  \num{7077888}x^{15}+ \num{41189376}x^{14}- \num{154386432}x^{13}+  \label{eq:f-delta}\\
  &\num{416788480}x^{12}- \num{857112576}x^{11}+ \num{1383417856}x^{10} - \num{1779354624}x^9+ \nonumber\\
  &\num{1834437632}x^8- \num{1514108928}x^7 + \num{992782336}x^6 - 
  \num{509312064}x^5 +   \nonumber\\
  &\num{199354208}x^4 -\num{57160752}x^3 + \num{11200088}x^2- \num{1313928}x + \num{67417} 
   \nonumber\\
       f_{\upperLeftPoint^x}(x) =& \num{16777216}x^{16}- \num{29360128}x^{14}+ \num{21757952}x^{12}- 
       \num{8978432}x^{10}+ \num{196608}x^9+  \label{eq:f-x}\\
       &\num{2187264}x^8- \num{208896}x^7  - 
       \num{233472}x^6+ \num{38400}x^5 - \num{2432}x^4+ \num{2304}x^3+  \nonumber\\
       &\num{1008}x^2- \num{648}x + \num{81}  
        \nonumber\\
       f_{\upperLeftPoint^y}(x) =& \num{16777216}x^{16}- \num{268435456}x^{15}+ \num{2009071616}x^{14}- 
       \num{9336520704}x^{13}  +   \label{eq:f-y}\\
       &\num{30152589312}x^{12}- \num{71751434240}x^{11}+ 
       \num{130119041024}x^{10}-    \nonumber\\
       & \num{183392632832}x^9+ \num{202951155712}x^8 - \num{176850272256}x^7+  \nonumber\\
       &\num{120867188736}x^6- \num{64057278976}x^5+ \num{25783384192}x^4 - \nonumber\\
       &\num{7610732416}x^3+ \num{1551687280}x^2 - \num{194938464} x + \num{11350269} 
        \nonumber
  \end{align}
\begin{figure}[t]
  \begin{subfigure}[b]{.5\linewidth}
    \centering
    \includegraphics[width=.66\linewidth, page=2]{figures/r_r_square.pdf}
    \caption{Any $0 \leq \delta \leq 1$ defines $\fociPoint,\leavePoint, \upperLeftPoint$ and ellipse $\ellipse$.}\label{fig:ellipse-2x2-square}
  \end{subfigure}%
  \begin{subfigure}[b]{.5\linewidth}
     \centering
      \includegraphics[width=.66\linewidth, page=1]{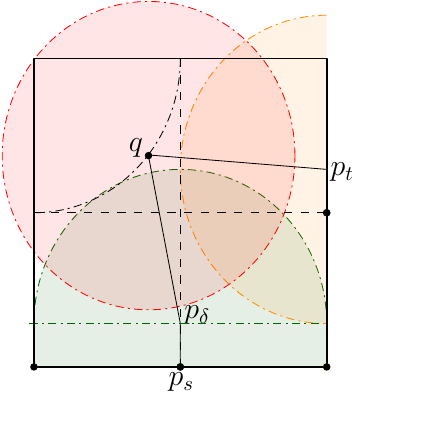}
      \caption{The optimal path $\optPath$ through $\sq$.}\label{fig:opt-2x2-square}
  \end{subfigure}
  \caption{Visualizations for \cref{lemma:optimal-path-square-with-circle}.}
\end{figure}
The value for $\delta\approx 0.167876$ defines the points $\fociPoint$ and $\leavePoint$. 
Together with the values for~$\upperLeftPoint^x,\upperLeftPoint^y$, we obtain the path above.
The combined length of the path is $\delta + \ellipsePathLength \approx 1.308838224$. 
%Note that the y position of the point $q$ is slightly above point $p_2$, see \cref{fig:opt-2x2-square}.
As~$\optPath$~contains a subpath that crosses the full height of $\quadrantBottomLeft$ and another subpath 
that crosses the full width of $\quadrantTopRight$, both quadrants are covered by $\optPath$, see \cref{fig:opt-2x2-square}. 
By construction, the bottom right quadrant is covered by the segment $\enterPoint\fociPoint$ and the point $\leavePoint$. 
The top left quadrant is covered by $\upperLeftPoint$, because $\quadrantTopLeft$ is fully contained in a disk with diameter $1$ 
centered in $\upperLeftPoint$. 
Therefore, $\optPath$ is a feasible path between two adjacent edges of $\sq$ with a length of~$L \approx 1.309$.
\end{proof}

\begin{restatable}{lemma}{optTourInSquare}\label[lemma]{lemma:optimal-tour-in-square}
A square $\squarePol$ of side length $2$ has a uniquely-shaped optimal Lawn Mowing tour~$\optTourSquare$ of length $L = 4L_{\sq}$, where $L_{\sq} \approx 1.309$.
\end{restatable}
\begin{proof}
  We start by subdividing $\squarePol$ by its vertical and horizontal center line into four quadrants (squares) $\sqnoindex_0,\dots,\sqnoindex_3$ with side length $1$.
  To cover the center point of each quadrant, a Lawn Mowing tour has to intersect it at least once.
  As $\squarePol$ is convex, $\optTourSquare$ cannot leave $\squarePol$ at any point.
  Finally, $\optTourSquare$ is symmetric with respect to the vertical and horizontal lines because otherwise, the quadrant subpaths 
  could be replaced by the shortest one.
  By \cref{lemma:optimal-path-square-with-circle}, there is a unique optimal Lawn Mowing path through each quadrant 
  yielding an optimal tour of length $L = 4 L_{\sq} \approx 4\cdot 1.309 \approx 5.235$, see \cref{fig:opt-4x4-square}.
\end{proof}
\subsection{Galois Group of the Polynomial}
Now we show that the coordinates of the optimal path $\optPath$ can not be expressed by radicals. We~employ a similar technique as
Bajaj~\cite{bajaj1988algebraic} for the generalized Weber problem.
A \emph{field}~$K$ is said to be an \emph{extension} 
(written as $K/\Q$) of $\Q$
if $K$ contains $\Q$.
Given a polynomial~$f(x) \in \Q[x]$, a finite extension $K$ of $\Q$ is a \emph{splitting field} over $\Q$ for $f(x)$ 
if it can be factorized into linear polynomials~$f(x) = (x-a_1)\cdots(x-a_k)\in K[x]$ but not over any proper subfield of~$K$. Alternatively, $K$ is a splitting field of $f(x)$ of degree $n$ 
over $\Q$ if~$K$ is a minimal extension of~$\Q$ in which~$f(x)$ has $n$ roots.
Then the \emph{Galois group} of the polynomial $f$ is defined as the Galois group of~$K/\Q$.
In principle, the Galois group is a certain permutation group of the roots of the polynomial.
From the fundamental theorem of Galois theory, one can derive a condition for solvability by radicals of the roots of $f(x)$ in terms of algebraic properties of its Galois group.
We state three additional theorems from Galois theory and Bajaj's work. 
The~proofs can be found in~\cite{bajaj1988algebraic,herstein1991topics}.

\begin{lemma}[\cite{herstein1991topics}]\label[lemma]{lemma:galois-not-solvable}
$f(x)\in \Q[x]$ is solvable by radicals over $\Q$ iff the Galois group over $\Q$ of $f(x)$, $Gal(f(x))$, is a solvable group.
\end{lemma}
\begin{lemma}[\cite{herstein1991topics}]\label[lemma]{lemma:sn-not-solvable}
The symmetric group $S_n$ is not solvable for $n\geq 5$.
\end{lemma}
\begin{lemma}[\cite{bajaj1988algebraic}]\label[lemma]{lemma:bajaj-sn}
If $n \equiv 0 \mod 2$ and $n>2$ then the occurrence of an $(n-1)$-cycle, an $n$-cycle, and a permutation of type $2+(n-3)$ on factoring the polynomial $f(x)$ modulo primes that do not divide the discriminant of $f(x)$ establishes that $Gal(f(x))$ over $\Q$ is the symmetric group $S_n$.
\end{lemma}

%For our proof, we use the software package Magma~\cite{MR1484478}. It provides
%a mathematically rigorous environment for working with algebraic
%structures such as groups, fields, and rings.

\begin{proof}[Proof of \cref{theorem:not-solvable-by-radicals}]
%For showing that the polynomials in \cref{eq:f-delta,eq:f-x,eq:f-y} are not solvable by radicals over $\Q$,
It suffices to show that $f_\delta$ is not solvable by the radicals as it describes the y-coordinates of two points in the solution. We provide three factorizations of $f_\delta$ modulo three primes that do not divide the discriminant $disc(f_\delta(x))$.
\begin{align}
  f_\delta(x)\equiv \;& \num{12} (x^{16}+\num{11} x^{15}+\num{20} x^{14}+\num{20} x^{13}+\num{12} x^{12}+\num{15} x^{11}+\num{20}
  x^{10}+\num{22} x^9+\num{19} x^8+ \num{2} x^7+ \nonumber\\
  &\num{18} x^6+\num{10} x^5+\num{12} x^4+\num{19} x^3+\num{16} x^2+\num{9} x+\num{8}) \mod 23\nonumber\\
  f_\delta(x)\equiv \;& \num{21} (x+\num{44}) (x^2+\num{34} x+\num{39}) (x^{13}+\num{4} x^{12}+x^{11}+\num{41} 
  x^{10}+\num{12} x^9+\num{21} x^8+\num{24} x^7+\nonumber\\
  &\num{32} x^5+\num{22} x^4+\num{10} x^3+\num{24} x^2+\num{18} x+\num{13}) \mod 47\nonumber\\
  f_\delta(x)\equiv \;& (x+\num{39}) (x^{15}+\num{8} x^{14}+\num{43} x^{13}+\num{23} x^{12}+\num{19} x^{11}+
  \num{38} x^{10}+\num{9} x^9+\num{6} x^8+\num{17} x^7+\nonumber\\
  &\num{34} x^6+\num{46} x^5+\num{43} x^4+\num{27} x^3+\num{50} x^2+\num{56} x+\num{1}) \mod 59 \nonumber
\end{align}

For the good primes $p=23, 47,$ and $59$, the degrees of the irreducible factors of $f_\delta(x) \mod p$ gives us an $16-cycle$, a $2+13$ permutation and a $15$-cycle, which is enough to show with \cref{lemma:bajaj-sn} and $n=16$ that $Gal(f_\delta)=S_{16}$.
By \cref{lemma:sn-not-solvable}, $S_{16}$ is not solvable; with \cref{lemma:galois-not-solvable}, this proves the theorem.
\end{proof}

\newcommand{\polyomino}{\Pol}
\newcommand{\boundingBox}{B}
\newcommand{\tiles}{T}
\newcommand{\tilesInside}{T_I}
\newcommand{\tilesInsideCopy}{T_I'}
\newcommand{\tilesOutside}{T_O}
\newcommand{\tile}{\sqnoindex}
\newcommand{\tilePath}{\optPath}
\newcommand{\connectionPoint}{c}
\newcommand{\cPBottom}{\connectionPoint_b}
\newcommand{\cPRight}{\connectionPoint_r}
\newcommand{\cPLeft}{\connectionPoint_l}
\newcommand{\cPTop}{\connectionPoint_t}
%
%Approximation proof
\newcommand{\hexagonTour}{\T_H}
\newcommand{\optTour}{\T_{OPT}}
\newcommand{\hexagons}{H}
\section{Mowing Polyominoes}
\label{sec:polyomino}

In the following, we analyze good tours for \emph{polyominoes},
which naturally arise when a geometric (or geographic) 
region is mapped, resulting in axis-parallel edges and integer vertices.
%In the following, we describe an overall approach for analyzing
%and computing good trajectories for this class of regions.
In the subsequent two sections, we describe the ensuing theoretical
worst-case guarantees (\cref{sec:approx}) and the
practical performance (\cref{sec:experiment}).

\subsection{Combinatorial Bounds}
For a unit-square cutter, the LMP on polyominoes
naturally turns into the TSP on the dual grid graph induced by pixel
centers.

\begin{lemma}
\label{lem:area}
Let $N\geq 2$ be the area of a polyomino $\Pol$ to be 
mowed with a unit-square cutter, and let $L$ be
the minimum length of a Lawn Mowing tour. Then $L\geq N$.
In the case of a unit-square cutter, $L=N$ iff the dual grid graph
of $\Pol$ has a Hamiltonian cycle.
\end{lemma}

This follows from Lemma 2 in the paper by Arkin~et~al.~\cite{Arkin2000} (which argues that there is an optimal LMP tour for a polyomino whose vertices are pixel centers) and implies the NP-hardness of the LMP~(Theorem~1~in~\cite{Arkin2000}).
%In their highly cited paper, Arkin et al.~\cite{Arkin2000}
%provided a number of corresponding combinatorial upper bounds. 
In particular,
they focused on grid graphs without a \emph{cut vertex}, which
is a node $v$ whose removal disconnects $G$:
``If $G$ has a cut vertex $v$, then we can consider separately the approximation
problem in each of the components obtained by removing $v$, and then splice the
tours back together at the vertex $v$ to obtain a tour in the entire graph~$G$.
Thus, we concentrate on the case in which $G$ has no cut vertices.''

For a simply connected polyomino consisting of $N$ pixels, 
the corresponding grid graph $G$
does not have any \emph{holes}, i.e., the complement of $G$ in the
infinite integer lattice is connected. These allow a tight combinatorial
bound on the tour length.
If $G$ has no cut vertices, then a combinatorially bounded tour of $G$ exists, as noted by Arkin~et~al.~\cite{Arkin2000} as follows. 

\begin{theorem}[Theorem 5 in \cite{Arkin2000}]
\label{th:simple_gg}
Let $G$ be a simple grid graph, having $N$ nodes at the centerpoints, $V$, 
of pixels within a simple rectilinear polygon, $R$, having $n$ (integer-coordinate) sides. 
Assume that $G$ has no cut vertices.
Then, in time $O(n)$, one can find a representation of a tour, 
$T$, that visits all $N$ nodes of $G$, of length at most $\frac{6N-4}{5}$.
\end{theorem}

For polyominoes with holes, there is a slightly worse, but still relatively
tight combinatorial bound of $\frac{53N}{40}=1.325N$ for the tour length, as follows.

\begin{theorem}[Theorem 7 in \cite{Arkin2000}]
\label{th:nonsimple_gg}
Let $G$ be a connected grid graph, having $N$ nodes at the centerpoints, $V$, 
of pixels within a (multiply connected) rectilinear polygon, $R$, having $n$ 
(integer-coordinate) sides. Assume that $G$ has no local cut vertices. 
Then, in time $O(n)$, one can find a representation of a tour, $T$, that visits 
all $N$ nodes of $G$, of length at most $1.325N$.
\end{theorem}

\subsection{Mowing with a Disk}
The natural lower bound of \cref{lem:area} still applies when mowing
with a circular cutter, because any unit distance covered
by the cutter can at most cover a unit area. However, meeting
(or approximating) this bound is no longer possible by simply 
finding a Hamiltonian cycle (or a good tour) in the underlying 
grid graph, as a circular cutter may cover already mowed area
or area outside of $\Pol$ when dealing with pixel corners.
Minimizing this effect ultimately leads to the 
algebraic analysis from the previous section.

\begin{figure}[t]
    \begin{subfigure}[b]{.49\linewidth}
    \centering
    \includegraphics[width=.9\textwidth]{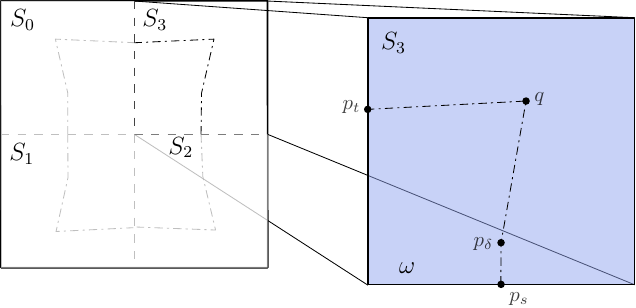}
    \vspace{2.35cm}
    \caption{Optimal Lawn Mowing tour for a $2\times 2$ square.}\label{fig:opt-4x4-square}
    \end{subfigure} %
    \begin{subfigure}[b]{.49\linewidth}
    \centering
    \includegraphics[width=.9\textwidth]{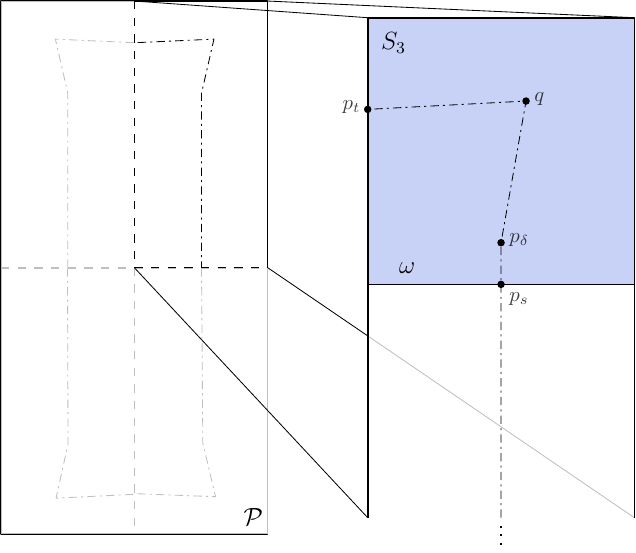}
    \caption{Optimal Lawn Mowing tour for a rectangle.}\label{fig:opt-4xh-rectangle}
    \end{subfigure}
    \caption{Optimal Lawn Mowing tours for a square and a rectangle.}
    \label{fig:opt}
\end{figure}
%

%In the following, we show how the underlying local
%insights can be exploited for good overall tours.
A starting point for further insights is illustrated in \cref{fig:opt}:
The optimal path from \cref{lemma:optimal-path-square-with-circle} with length $L_{\sq}$ can be used
for rectangles with width $2$ and arbitrary height $h\geq
2$.
%To this end, we extend the path from $\enterPoint$ outwards perpendicular to
%the $1\times 1$ square $\sq$.  One can use a similar construction as in
%\cref{lemma:optimal-tour-in-square} to obtain optimal tours for arbitrary
%rectangles, refer to \cref{fig:opt-4xh-rectangle}.
%
\begin{corollary}
Any rectangle $\Pol$ with width $2$ and height $h>2$ has a uniquely-shaped optimal Lawn Mowing tour $\optTourSquare$ of length $L=4L_{\sq} + 2h - 8$.
\end{corollary}

Extending this idea to more general polyominoes leads to
realizing a tour of the dual grid with locally optimal
``puzzle pieces'': a limited set of locally good trajectories 
that mow each visited pixel, which are merged at
transition points on the pixel boundaries;
%As it turns out, the structure of each local trajectory 
%depends only on the turn angle at a tile; 
see \cref{fig:tile-overview}.
The construction of the puzzle pieces is done in \cref{sec:region-paths}.

%The first part is the construction of grid center points that define the regions. Grid points may adapt to the polygon shape so that the cutter can cut straight-line segments across the boundary. The second step is to compute the grid graph and regions for every center point. The last step is to solve a modified turn cost IP to optimality which yields a feasible tour for lawn mowing/milling.
\begin{figure*}
  \begin{subfigure}{.24\linewidth}
    \includegraphics[width=\linewidth,page=1]{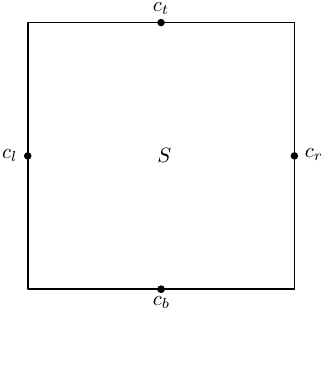}
    \caption{A pixel $\tile$ with transition points.}\label{fig:tile-overview}
  \end{subfigure}
  \begin{subfigure}{.24\linewidth}
    \includegraphics[width=\linewidth,page=2]{figures/tile_paths.pdf}
    \caption{A straight-line covering, with zero turn}\label{fig:tile-straight-line}
  \end{subfigure}
  \begin{subfigure}{.24\linewidth}
    \includegraphics[width=\linewidth,page=3]{figures/tile_paths.pdf}
    \caption{A covering for a single, 90-degree turn}\label{fig:tile-right-turn}
  \end{subfigure}
  \begin{subfigure}{.24\linewidth}
    \includegraphics[width=\linewidth,page=4]{figures/tile_paths.pdf}
    \caption{A covering for a double, 180-degree U-turn}\label{fig:tile-u-turn}
  \end{subfigure}
  \caption{A pixel $\tile$ with three elementary covering trajectories.}\label{fig:tile_covering_trajectories}
\end{figure*}

%Finding a solution to the puzzle that connects adjacent pieces so that the paths match at the connection points can be formulated as an integer program.
%Note that it might be possible that a region is traversed multiple times or that puzzle pieces outside the polygon are necessary to construct shorter tours.

\subsection{Constructing Puzzle Pieces}\label{sec:region-paths}
\old{
To construct the puzzle pieces, we compute the bounding box $\boundingBox$ of
$\polyomino$ and subdivide it into axis-aligned $2\times2$ tiles.  A tile
$\tile$ can either be fully contained in $\polyomino$ or outside of
$\polyomino$. This yields two sets of tiles $\tilesInside, \tilesOutside$ for
the inner and outer tiles, respectively.
}

%In the following, we describe such local trajectories.
In order to analyze locally good trajectories for mowing visited pixels, consider
the four corners of a pixel with coordinates $(0,0), (1,0), (1,1), (0,1)$.
%To ensure that 
%We will calculate several paths through $\tile$ that cover $\tile$ completely. To ensure that 
%all local trajectories line up to form an overall feasible solution, 
We consider \emph{transition points} $\cPBottom=(1/2,0)$, $\cPRight=(1,1/2)$, $\cPTop=(1/2,1)$, and $\cPLeft=(0,1/2)$
at the edge centers to ensure an overall
connected trajectory, as shown
in \cref{fig:tile-overview}.
%Subject to symmetry and rotation, 
There are three combinatorially distinct ways
for visiting a pixel, corresponding to \cref{fig:tile-straight-line,fig:tile-right-turn,fig:tile-u-turn}.
These are (i) a straight path, (ii) a simple turn, and (iii) a U-turn.
\begin{enumerate}[i]
    \item The straight path connects $\cPBottom$ and $\cPTop$ and has length $L_\tilePath = 1$, see~\cref{fig:tile-straight-line}.
    \item The simple turn connects $\cPBottom$ and $\cPRight$. 
    Solving a minimization problem similar to the one from the proof of \cref{lemma:optimal-path-square-with-circle}
    with a fixed $\delta=0$ yields a path $\cPBottom, q, \cPTop$ of length $L_\tilePath \approx 1.32566$ and
    $q = (\frac{1}{2\sqrt{2}}, \frac{1}{4}(4-\sqrt{2}))$, see~\cref{fig:tile-right-turn}.
    \item The U-turn connects $\cPBottom$ with itself while covering $\tile$ completely. 
    An optimal solution must visit both circles of unit diameter 
    centered at $(0,1), (1,1)$. Thus, we can formulate the following minimization problem
    \begin{align*}
    &\text{min}  && \sqrt{(x_1-\frac{1}{2})^2 + y_1^2} + \sqrt{(x_2-\frac{1}{2})^2 + y_2^2} + \sqrt{(x_1-x_2)^2 + (y_1-y_2)^2\vphantom{\frac{1}{2}}} &\\
    &\text{s.t.} && x_1^2+(y_1-1)^2 - \frac{1}{4} &=~ 0\\
    &&& (x_2 - 1)^2+(y_2- 1)^2 - \frac{1}{4} &=~ 0
    \end{align*}
    This yields an optimal solution $\cPBottom, q_1,q_2, \cPBottom$ with 
    $q_1\approx (0.383, 0.678)$, 
    $q_2 \approx (0.617, 0.678)$ and length $L_\tilePath \approx 1.611183$.
\end{enumerate}
Note~that we do not use the optimal path from \cref{lemma:optimal-path-square-with-circle}, because 
it uses transition points that are slightly off center, $\leavePoint \neq \cPRight$, 
with the imbalance canceled out between two adjacent simple turns.
Thus, using central transition points incurs a small
marginal cost when compared to an optimal trajectory
($1.32566$ vs. $1.309$, or about 1.2\% longer for each simple turn), 
but it sidesteps the higher-order difficulties of combining longer off-center
strips.

%All remaining paths between other combinations of connection points can be built by rotating or flipping paths (i)-(iii). Tiles not needed for coverage can be traversed using direct paths between the connection points. This will allow shortcuts in tours that need to leave $\polyomino$.

\subsection{Building an Overall Tour}\label{sec:overall}

%The path generation procedure from \cref{sec:region-paths} yields a set of tiles $\tiles=(\tilesInside, \tilesOutside)$ together with a set of paths.

Making use of the puzzle pieces, we can now approach the LMP in three steps,
as follows.

\begin{enumerate}
\item[\bf{A}] Find a cheap roundtrip on the dual grid graph.
\item[\bf{B}] Carry out the individual pixel transitions based on the above
puzzle pieces as building blocks to ensure coverage of all pixels and thus a feasible tour.
\item[\bf{C}] Perform post-processing sensitive to the transition costs on the resulting tour to achieve further improvement.
\end{enumerate}

In the following sections, we describe how the involved steps 
can be carried out either with an emphasis on worst-case runtime
and worst-case performance guarantee
(giving rise to theoretical approximation algorithms, as discussed
in the following \cref{sec:approx}), or with the goal of good
practical performance in reasonable time for a suite of benchmark instances
(leading to the experimental study described
in \cref{sec:experiment}). 

\old{
a graph problem on a grid-like graph; as the cost for transitioning
a pixel in a manner that fully covers it depends on the type of \emph{covering}
turn, this amounts to solving 
a tour problem with both distance and turn cost.

To this end, let $H=(\tiles, E_\tiles)$ be the dual grid graph. 
Any tour in $H$ that visits all pixels of $\tilesInside$ can be transformed
into a covering tour for $\polyomino$ by replacing an adjacent triplet
$\tile_u,\tile_v,\tile_w$ with the appropriate path between
$\connectionPoint_{uv}$ and $\connectionPoint_{vw}$.  A tour may visit pixels 
$\tile_v$ in $\tilesInside$ multiple times. On the second traversal of
$\tile_v$, the tour can use the direct path between the transition points
$\connectionPoint_{uv}$ and $\connectionPoint_{vw}$ instead.
For this purpose, we create a copy $\tile_v'$ for each pixel
in $\tilesInside$. We denote the set of all copied pixels by $\tilesInsideCopy$.
For every edge $\{\tile_u,\tile_v\} \in E_\tiles$, we insert an additional edge
$\{\tile_u,\tile_v'\}$ if $\tile_v'\in \tilesInsideCopy$,
$\{\tile_u',\tile_v\}$ if $\tile_u'\in \tilesInsideCopy$ and
$\{\tile_u',\tile_v'\}$ if $\tile_u',\tile_v'\in \tilesInsideCopy$.  Thus any
tile $\tile_v'\in \tilesInsideCopy$ will have the same neighbors in
$\tilesInside$ and $\tilesInsideCopy$ as $\tile_v$. Intuitively, adjacent pixels
remain adjacent, and the neighborhood is extended by the duplicate pixels of the
previous neighbors.
Let $G=(T\cup \tilesInsideCopy, E)$ be the graph resulting from this procedure.

Finding a good tour for $\polyomino$ that uses the subtrajectories from
\cref{sec:region-paths} now corresponds to finding a shortest tour through
$G$.\todo{It is more complicated than this} 

\todo{Discuss: Distinguish between theoretical, approximate solution, and
practical, IP solution}

\todo[inline]{Add IP in experiments!}
}

%\input{./03b_solver.tex}

%The cost of traversing a pixel
%$\tile_v$ is dependent on its neighbors in the tour.
%This property of the tour problem is very similar to the TSP with Turncost, for which various algorithms exist.
\old{
We derive an integer program for connecting the previously generated paths in an optimal fashion, see~\cite{krupkediss} \todo{Is this the best reference?} for details.
The integer program uses non-negative variables $x_{uvw}=x_{uvw}$ for a tile $\tile_v \in T\cup \tilesInsideCopy$ and adjacent tiles $\tile_u,\tile_w \in N(\tile_v)$.
Each variable states how often a transition is used in the solution.
Constraints ensure that every tile in $\tilesInside$ is covered by some transition of the tour.
The formulation guarantees that the resulting transitions line up to form closed cycles.
Additional subtour constraints are added dynamically during the solving process to enforce a single cycle, i.e., a tour in $G$.
}

\section{Theoretical Performance: Approximation}
\label{sec:approx}

%In the following, we realize the basic algorithmic setup
%of \cref{sec:overall} with the goal of constant-factor
%approximation. 
For constant-factor approximation, 
we start with a low-cost roundtrip in the dual grid graph (Step A),
e.g., with the previous results of Arkin et al.~\cite{Arkin2000}.
Step B is realized using the puzzle pieces of \cref{sec:region-paths}
for a feasible tour,
at a cost of $1+\tau:=1.32566$ for each 90-degree turn in the grid tour
(corresponding to piece (ii));
note that the turn cost for a U-turn of $1.61118$ 
(corresponding to piece (iii)) does not exceed $1+2\tau$. 
%We can then achieve an overall approximation algorithm with
%a low factor 
By using combinatorial arguments for the post-processing Step C, we can 
prove that a limited number of covering turns (with an additional turn cost $\tau$) 
suffices for overall feasibility.\\
%, while other turns 
%can be carried out without incurring $\tau$.

\begin{theorem}
\label{th:tau}
Let $\Pol$ be a polyomino with $N>5$ pixels, and let $T$
be a tour of the dual grid graph of length $L$.
Then we can find a feasible Lawn Mowing tour
for a unit-diameter disk of length at most $L(1+\tau)$. 
\end{theorem}

\begin{proof}
Let $T$ be a tour of the dual grid graph; let $L$ be the length of $T$.
$L$ is the total number of visits of individual pixels, inducing the following three categories
of pixel visits.

\begin{enumerate}
\item $L_0$ ``free'' visits of pixels, in which no covering turn occurs, and no turn cost is incurred.
\item $L_1$ ``one-turn'' visits of pixels, in which one covering turn occurs, for a turn cost of $\tau$.
\item $L_2$ ``U-turn'' visits of pixels, in which a double covering turn occurs, for a turn cost of not more than $2\tau$.
\end{enumerate}

Let $p_i$ be a pixel that is visited in step $i$ of the tour by a U-turn of $T$.
Then $p_i$ is adjacent to a pixel $q=p_{i-1}=p_{i+1}$ that was left in step $i$ and
entered in step $i+1$. Because no pixel visited by a U-turn needs to be visited 
more than once, as well as $N>5$, the pixel $q$ cannot only have neighbors 
that are visited by U-turns. Therefore, $q$ has a predecessor in the
tour that is not a U-turn, (w.l.o.g., $p_{i-2}$); this visit from $p_{i-2}$
is either a one-turn visit with a covering turn, or a free visit. In either case,
$q$ is already covered when visited from $p_i$, 
and we can simply follow the grid path at only the distance cost of 1. 

As a consequence, each U-turn visit (incurring a cost not exceeding $2\tau$)
can be uniquely mapped to a free visit of its successor (incurring no turn cost), 
and the overall cost for all covering turns does not exceed $L\tau$, 
for a total length of at most $L(1+\tau)$, as claimed.
\end{proof}

\old{
\begin{figure}[t] 
    \centering
    \includegraphics[width=.4\textwidth]{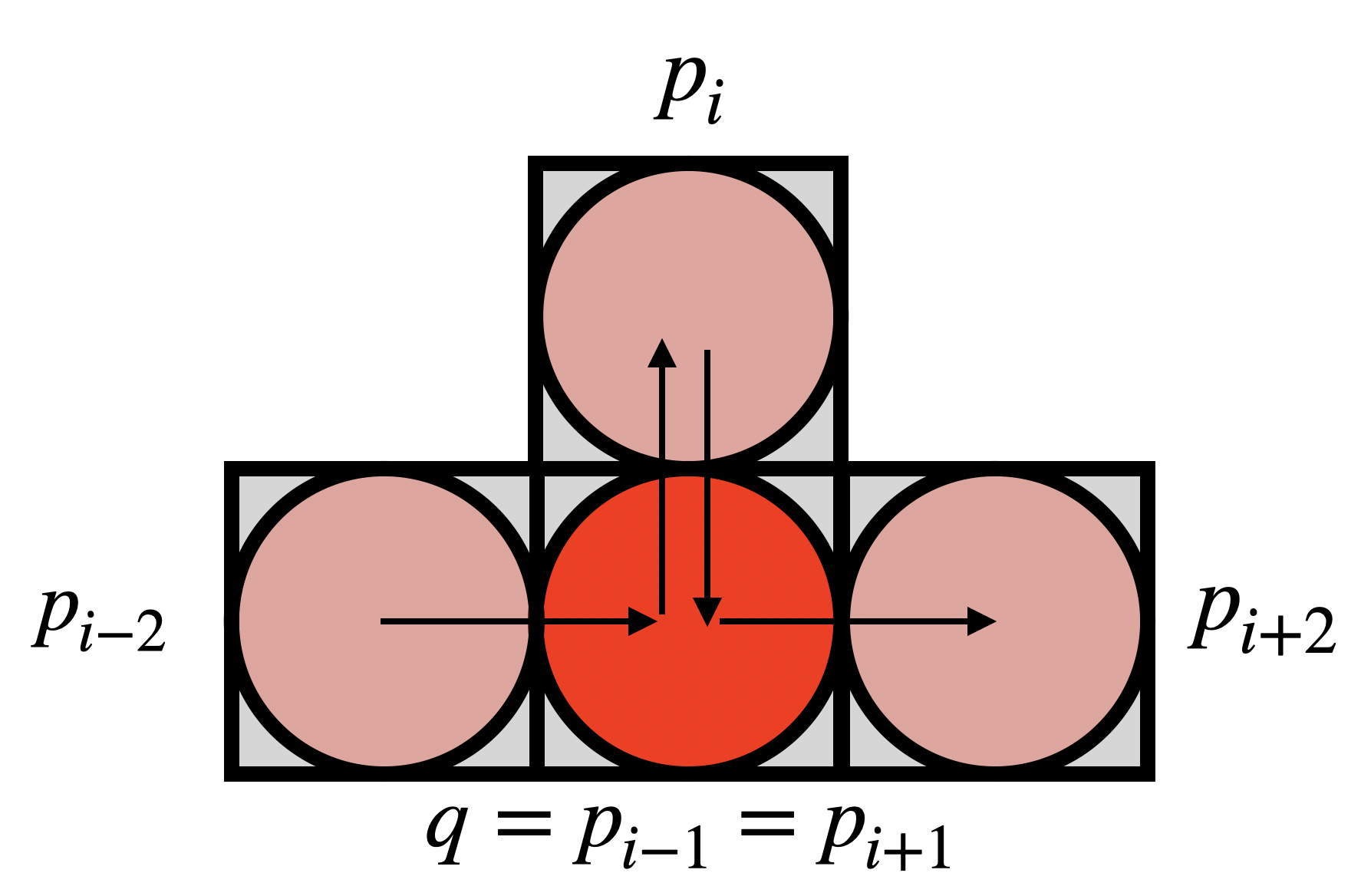} 
    %\vspace{2.35cm}
    \caption{Mapping the turn cost for a covering U-turn to a free visit of the successor pixel.}\label{fig:doubleu}
    \label{fig:turncount}
\end{figure}
}

%\todo{Include actual figure.}

For simple polyominoes without cut vertices,
\cref{th:simple_gg} provides a tour $T$ in the dual grid graph
of length at most $\frac{6N-4}{5}$, implying the following.

\begin{corollary}
\label{th:simple_disc}
Let $\Pol$ be a simple polyomino with $n$ vertices and $N$ pixels, 
whose dual grid graph does not have any cut vertices. 
Then, in time $O(n)$, one can find a representation of a feasible Lawn Mowing trajectory $T$
for a unit-diameter disk of length at most $\frac{6N-4}{5}\tau$, which is
within $1.5908$ of the optimum.
\end{corollary}

For polyominoes with holes, 
we can apply the same line of argument
to a tour $T$ of the dual grid graph obtained 
from \cref{th:nonsimple_gg}.

\begin{corollary}
\label{th:nonsimple_disc}
Let $\Pol$ be a (not necessarily simple) polyomino with $n$ vertices and $N$ pixels,
whose dual grid graph does not have any cut vertices.
Then, in time $O(n)$, one can find a representation of a feasible Lawn Mowing trajectory $T$
for a unit-diameter disk of length at most $\frac{53N}{40}\tau$, which is
within $1.7565$ of the optimum.
\end{corollary}

As the number of turns is of critical importance for the overall
cost of a Lawn Mowing tour obtained from a tour of the dual grid graph,
we can %try to minimize the required number of turns
consider optimizing a linear combination of
tour length and turn cost.  Arkin et al.~\cite{arkin2005optimal}
gave a PTAS for this problem, as follows.

\begin{theorem}[Theorem 5.17 in \cite{arkin2005optimal}]
\label{th:ptas}
Define the cost of a tour to be its length plus $C$ times the number of (90-degree) turns. 
For any fixed $\varepsilon > 0$, there is a $(1 + \varepsilon)$-approximation 
algorithm, with running time $2^{O(h)} N^{O(C)}$, 
for minimizing the cost of a tour for an integral orthogonal polygon~$\Pol$ with $h$ holes and 
$N$ pixels.
\end{theorem}

Combining tour length and turns 
allows providing more explicit bounds, as follows.
Additional local considerations are possible, but these do not 
necessarily improve the worst-case bounds. Instead, they are
employed heuristically in the practical section.

\begin{theorem}
\label{th:onlyturn}
Let $\Pol$ be a polyomino with $n$ vertices and $N$ pixels,
and let $T$ be a tour of the dual grid graph of length $L$
and a total of $t$ (weighted) turns. Then there is a feasible
Lawn Mowing tour of cost at most $L+t\tau$.
\end{theorem}

%This achieves a better bound whenever $t<L$; in particular,
%we get better approximation factors for polyominoes $\Pol$ in 
%which the number of pixels is large compared to the number of polygon edges,
%i.e., $n<<N$, because then the number of turns in a tour
%of the dual grid graph remains small compared to $N$.

%\begin{corollary}
%\label{th:fat}
%Let $\Pol$ be a polyomino with $n<<N$ vertices and $N$ pixels.
%If there is a tour of length $cN$ for 
%Then we for any $\varepsilon>0$, we can construct within polynomial
%time a feasible Lawn Mowing tour of cost 
%within $(1+\varepsilon)$ of an optimal tour.
%\end{corollary}
%
%\todo{This may be wishful thinking; I guess this only works if 
%we do not have any cut vertices.}
%The idea is to use a good tour $T$, and observe that at it 
%contains $t$ turns.

%\newcommand{\tsp}{\textup{TSP-Base}}
\newcommand{\tsp}{\textup{TSP$_{\text{Small}}$}}
\newcommand{\tspCoverage}{\textup{TSP$_{\text{Cov}}$}}
\newcommand{\tspTurncost}{TSP$_{\text{Turn}}$}

\newcommand{\dualGraph}{H}
\newcommand{\dualGraphSmall}{H'}
\newcommand{\dualGraphPixels}{V}
\newcommand{\dualGraphPixelsInside}{V_\Pol}
\newcommand{\dualGraphEdges}{E_H}
\newcommand{\turncostCycle}{O}

\section{Practical Performance: Algorithm Engineering}
\label{sec:experiment}

\subsection{Algorithmic Tools}
Here we exploit the algorithmic
approach of \cref{sec:overall} for good \emph{practical}
performance for general polygonal regions,
starting with a preprocessing step: 
%\begin{enumerate}
For a given polygonal region~$Q$, find a suitable
polyomino $\Pol$ that covers it.
%\end{enumerate}

%Subsequently, %This is followed by finding a good solution for mowing $\Pol$, for 
We can then aim for practical minimization of tour length and turn cost
for \textbf{A} (analogous to the theoretical \cref{th:ptas}), and
use puzzle pieces in \textbf{B} for a feasible
tour. In principle, we can approach \textbf{A} 
by considering an integer program (IP); however, solving this IP
becomes too costly for larger instances, so
we use a more scalable approach: \textbf{(A)} Find 
a good TSP solution on the dual grid graph; \textbf{(B)} insert puzzle pieces;
\textbf{(C)} minimize the induced turn cost by Integer Programming and Large Neighborhood Search (LNS).

%\subsection{Optimization Tools}
%
%the puzzle pieces in an optimal fashion. We call this solver \tspTurncost. As a baseline we compare we chose
%the proposed technique of Arkin~et~al.\cite{Arkin2000}, namely to consider a grid of size
%$\frac{\sqrt{2}}{2}$ and find an (Euclidean) optimal tour between all centerpoints of polygon-intersecting tiles 
%using the TSP solver \emph{concorde}. We call this approach \tsp.
%To compare the differences in grid size we also consider \tspCoverage{} which is a (Manhattan) TSP solver on the 1-grid
%where neighbors in the optimal tour are replaced with a shortest path in the dual grid graph.
%
%
% Defining graph and trajectories
\subsubsection{Choosing a Suitable Grid}
Consider a non-degenerate polygonal region $Q$, and a minimal covering polyomino $\Pol$ of cell size $\ell$. Without loss of generality, $Q$ contains only pixels with a point of $Q$ in their interior; furthermore, we can assume that both an $x$- and a $y$-coordinate of a grid point coincide with a coordinate of $Q$. 
%For a non-degenerate polygonal region $Q$, %for which we want to choose a suitable 
%a minimal covering polyomino $\Pol$ of cell size $\ell$ %. Any such $\Pol$
%contains only pixels with a point of $Q$ in their interior and
%can be identified by the $x$- and $y$-coordinates 
%of a single lattice point $; furthemore, for non-degenerate polygons $Q$,
%$we can restrict our attention to covering polyominoes that 
%only contain pixels with a point of $Q$ in their interior.
%xXxThis allows us to consider $x_{min}, y_{min}$ as the minimum $x,y$-coordinates of the
%xXxvertices of $\Pol$; for a given $k$, we can test all corners
%xXx$\left(x_{min}-\frac{i}{k}\ell, y_{min}-\frac{j}{k}\ell\right)$ for $i,j \in
%xXx\{0,\dots,k-1\}$ and select the grid with the minimum number of grid cells that
%xXxintersect with $\Pol$. 
%(w.l.o.g.) has both an $x$- and a $y$-coordinate of a grid point coincide with 
%a coordinate of $Q$.
This limits the %For our experiments %in this section we set $k=10$, i.e.,
number of relevant grid positions to a quadratic number of choices, %) the grid position that minimizes the number of
%xXx best of \SI{100} possible grid positions to obtain a good covering polyomino $\Pol$,
%xXx based on the underlying grid $H$.
%This yields $\dualGraph=(\dualGraphPixels, \dualGraphEdges)$ be the dual grid graph, with 
%$\dualGraphPixelsInside = \{v\in \dualGraphPixels \mid v\cap \Pol\neq \emptyset\}$ 
from which one can choose the one with the smallest number 
of pixels contained in the resulting polygon $\Pol$.

\subsubsection{Minimizing Tour and Turn Cost}
%We can convert a tour $T$ of the dual grid graph to a feasible
%Lawn Mowing tour by replacing visits to pixels by elementary puzzle pieces.
%As discussed in \cref{sec:polyomino} and \cref{sec:approx},
%this results in an overall length that depends on tour length of $T$
%and total turn cost. Thus, it makes sense to start with a grid tour of minimum
%combined cost before inserting the puzzle pieces. 
Finding a covering tour of minimum combined tour length and turn cost can be 
formulated as an IP\@.
% DK (2023-06-30): Added remark that we minimize the ACTUAL costs, not just the number of turns.
% However, the IP originates from minimizes turn costs.
As the cost for each turn can be specified individually in this IP, we can
also minimize the final tour length directly instead of just approximating it based on the number of turns.
In principle, this IP can be solved with CPLEX~\cite{CPLEX}
or Gurobi~\cite{gurobi}; however, this fails when aiming 
for truly large instances. (Even without the length of the tour,
the turn-cost problem is notoriously difficult~\cite{fekete2019practical}.)
Thus, we have pursued an alternative
approach that starts with a cheap roundtrip on the dual grid graph
in which we ignore the turn cost.
% DK (2023-06-30): The previous ending did not make much sense, thus I wrote what we are actually doing.
We then use this IP as part of a Large Neighborhood Search (described in \cref{sec:post-processing}) to minimize the actual costs of this solution,
 and for computing lower bounds on the best possible solution based on puzzle pieces.

\paragraph{Formulating the Integer Program}
To formulate the integer program,
let $\widehat{uvw}$ with $uv, vw\in \dualGraphEdges$ be the puzzle piece covering the pixel $v$ and connecting $c_{uv}$ and $c_{vw}$, 
and $\overline{uvw}$ be the direct path between $c_{uv}$ and $c_{vw}$.
We call these tour elements (\emph{covering} and \emph{non-covering}) \emph{tiles}.
We use the variables $x_{\widehat{uvw}} \in \mathbb{B}, uv, vw \in \dualGraphEdges$ to denote which covering tile, 
i.e., puzzle piece, is used for $v\in \dualGraphPixelsInside$ is in the tour.
For simplicity, $x_{\widehat{uvw}}$ is also defined for $v\in \dualGraphPixels\setminus \dualGraphPixelsInside$, but fixed to $0$.
Analogously, we are using the variables $x_{\overline{uvw}} \in \mathbb{N}_0, uv,vw\in \dualGraphEdges$ to denote how often which 
non-covering tiles, i.e., direct paths, for $v\in \dualGraphPixels$ are used in the tour.
Because we may need to pass a pixel multiple times, this is an integer variable.

% objective and constraints (without subtour elimination)
Finding the shortest set of cycles that cover all pixels $t\in \dualGraphPixelsInside$ can be expressed as follows.
Enforcing a single cycle, i.e., tour, is done later by some more complex constraints that need additional discussion.
\begin{align}
  \min & \displaystyle\sum_{uv, vw \in \dualGraphEdges} ||\overline{uvw}||\cdot x_{\overline{uvw}}+||\widehat{uvw}||\cdot x_{\widehat{uvw}} \label{eq:mip:obj} \\
    \text{s.t.} & \displaystyle\sum_{u,w \in N(v)} x_{\widehat{uvw}} = 1 & \forall v \in \dualGraphPixelsInside \label{eq:mip:coverage}\\
    &  \begin{array}{l}\displaystyle 2\cdot (x_{\overline{wvw}}+x_{\widehat{wvw}}) + \sum_{n \in N(v), n\not=w}  (x_{\overline{nvw}}+x_{\widehat{nvw}})\\
    = 
    \displaystyle 2\cdot (x_{\overline{vwv}}+x_{\widehat{vwv}}) + \sum_{n \in N(w), n\not=v}  (x_{\overline{vwn}}+x_{\widehat{vwn}})\end{array} & \forall vw \in \dualGraphEdges\label{eq:mip:flow}\\
  & x_{\widehat{vwv}}\in \mathbb{B}, x_{\overline{vwv}}\in \mathbb{N}_0 & \forall uv, vw\in \dualGraphEdges
\end{align}
The objective function~(\ref{eq:mip:obj}) minimizes the sum of lengths of the used tiles (the length of a tile is denoted by $||\cdot ||$).
\Cref{eq:mip:coverage} enforces that every pixel $v\in \dualGraphPixelsInside$ that intersects the polygon~$P$ has one covering tile; $N(v)$ are the neighbors of $v$.
\Cref{eq:mip:flow} ensures that every tile has a matching incident tile on each end, i.e., connecting all tiles yields feasible cycles.

% subtour elimination
% This is the complex part
\paragraph{Subtour Elimination}
Next, we have to add constraints that enforce a single tour.
A simple, but insufficient, constraint is similar to the classical subtour elimination constraint of the 
Dantzig-Fulkerson-Johnson formulation~\cite{dantzig1954solution} for the Traveling Salesman Problem.
For every non-empty subset $S\subset \dualGraphPixels, S\not=\emptyset, \dualGraphPixels\not \subset S, \dualGraphPixels\setminus S \not= \emptyset$ 
that contains a real part of $\dualGraphPixelsInside$, there has to be some path leaving the set to connect to 
$\dualGraphPixels\setminus S$.
\begin{equation}
  \sum_{uv, vw \in \dualGraphEdges, v\in S, w\not \in S} x_{\overline{uvw}}+x_{\widehat{uvw}} \geq 1 
\end{equation}
Unfortunately, this is not sufficient as we can have cycles that cross but are not connected, e.g., 
for the tiles $\overline{uvw}$ and $\widehat{svt}$ with $\{u,w\}\cap \{s,t\}= \emptyset$.
While they share the same pixel $v$ in the grid graph, the paths themselves do not have to intersect.
We can also not expect them to be exchangeable as this may increase the objective.
Let $\turncostCycle$ be a cycle of tiles that cover only a real subset of 
$\dualGraphPixelsInside$, $E(\turncostCycle)$ denote the edges in the grid graph, and $\widehat{abc}\in \turncostCycle$ be a 
covering tile of $\turncostCycle$ with $b\in V(\turncostCycle)\cap \dualGraphPixelsInside$.
The following constraint now forces the path that covers $v$ to change and connect to exterior parts.
\begin{equation}
  \sum_{u,w\in N(b), \widehat{ubw}\not \in \turncostCycle} x_{\widehat{ubw}}
  + \sum_{vw \in E(\turncostCycle), u \in N(v), \overline{uvw}\not \in \turncostCycle, \widehat{uvw}\not \in \turncostCycle} (x_{\overline{uvw}}+ x_{\widehat{uvw}}) \geq 1
\end{equation}
This constraint is sufficient as it can be applied to any cycle that is covering only a subset of $\dualGraphPixelsInside$, but generally less efficient.

\subsubsection{Finding a Cheap Roundtrip and Ensure Coverage}\label{sec:cheap-roundtrip}
We consider two different methods for computing different initial tours.

\subparagraph{\tsp:} 
Previous authors~\cite{bormann2015new,murtaza2013priority,sharma2019optimal,zheng2010multirobot}
have suggested using a grid graph $\dualGraphSmall$ 
with smaller cell size $\ell=\frac{\sqrt{2}}{2}$ for covering $\Pol$, or simply assumed square-shaped tools.
This eliminates the need to consider any turn cost, as smaller pixels are 
covered when the cutter visits their centers. 
This yields \tsp, which we use as a baseline.
Because of the smaller grid size, this may result in 
double coverage when parallel unit strips 
suffice to cover the $\Pol$, for a worst-case overhead
of $\sqrt{2}-1$, or about \SI{41.4}{\percent}.
%This simplicity
%is the reason why this approach has been in practical applications and
%approximation algorithms.  \todo{Cite something here} In the following, we
%We refer to the solver that calculates an optimal (Euclidean) TSP tour on
%$\dualGraphSmall$ as \tsp.
\subparagraph{\tspCoverage:} 
As described in the preceding \cref{sec:polyomino}, we can
use a cheap tour for the grid graph $\dualGraph$ with cell size $\ell=1$,
and perform the puzzle piece modification.
This combined solver %that uses solves TSP on $\dualGraph$ and inserts the puzzle pieces in the described way
is called \tspCoverage. As shown in 
 \cref{sec:approx}, we can limit the worst-case overhead for performing
turns of \tspCoverage\ to $\tau=0.32566$ per length of the tour, or about \SI{32.6}{\percent}.
%by solving the TSP for Manhattan distances on the vertices $\dualGraphPixelsInside$.
%convert The resulting tour can be converted into a covering 
%tour for $\Pol$ by transforming two 
%consecutive edges $uv, vw\in \dualGraphEdges$ to an appropriate path between $c_{uv}$ and $c_{vw}$.
%We call this the \emph{modification} step.
%For every $v\in \dualGraphPixelsInside$, one such path has to be a 
%(covering) tile (cf.~\cref{fig:tile_covering_trajectories}), 
%all others can be direct connections.
%We pick the shortest tile from the list of candidates 
%and then connect the open ends with (possibly non-covering) straight-line segments.

\subsubsection{Improving the Tour}\label{sec:post-processing}

% =====================================================
% This text describes the mixed integer program to 
% compute optimal solutions in the grid. The IP is
% an adapted variant of the turn cost IP in the
% master's thesis and dissertation of Dominik.
% =====================================================
%\todo[inline]{What we are actually doing here is changing the objective. Unfortunately, this is a more difficult problem than the classical TSP. As it takes the TSP as start solution, it is, however, also some kind of post-processing.}
%\todo[inline]{This initial part may be redundant with the previous section.}
% defining variables for the IP
For a feasible tour from \tspCoverage{}, %we combine Large Neighborhood Search (LNS) with
%integer programming to improve the tour %under an objective function that includes
%for tour length and turn cost; this yields \tspTurncost{}.
%An optimal covering tour through $\dualGraph$ can be stated as an integer program which is unfortunately very 
%hard to solve in practice.
%Branch-and-bound methods commonly used for solving 
%integer programs benefit strongly if they
%start with a good initial solution, as this allows earlier pruning of suboptimal branches. 
%To this end, %Moreover, we can use Large Neighborhood Search 
we use an LNS-algorithm~\cite{pisinger2019large}, which iteratively fixes a large part of the IP and only optimizes a small region of tiles; this yields \tspTurncost{}.
%and potentially improves the solution on the free part.
We select a random tile from the current tour and a fixed number of adjacent pixels.
% around it.
This yields a limited-size integer program, in which only the involved puzzle pieces are allowed to change.
%The variables of all other pixels can be replaced by constants, resulting in an integer program of fixed size.
%This fixed size is to be chosen such that the integer program can be solved quickly and we can perform many iterations.
%We tune the region size after each iteration, based on the runtime of the previous iteration to escape local minima, while limiting the runtime of iterations.
To escape local minima, we tune the size (and runtime) of the IP after each iteration based on the runtime of the previous iteration.
In the end, we attempt to solve the IP on the complete instance,
using the start solution from the LNS\@.
This provides lower bounds on the best placement of puzzle pieces. 
%Note that
%these bounds are not valid for the LMP in general as all tiles are restricted to the given grid
%graph $\dualGraph$.
%As this approach is relatively expensive, we only apply it to the initial solution.

\begin{figure}[h!]
    \centering
  \begin{subfigure}[b]{.5\linewidth}
    \includegraphics[width=\linewidth]{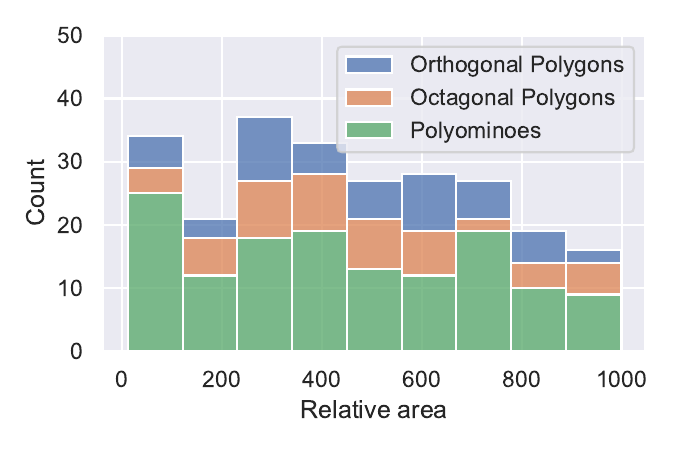}
    \caption{Distribution of instances.}
  \end{subfigure}%
  \begin{subfigure}[b]{.5\linewidth}
    \includegraphics[width=\linewidth]{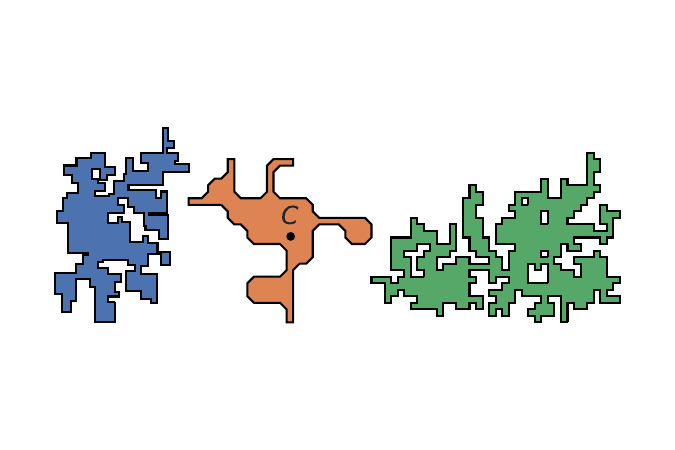}
    \caption{Three instances with the cutter $C$.}
\end{subfigure}
    \caption{Examples of the used polygons and their size distribution.}
    \label{fig:distribution}
\end{figure}

\subsection{Experimental Setup}
%Now we describe our practical implementation along with our experimental study for demonstrating the practical
%usefulness of our proposed algorithm.
Our practical implementation was tested on a workstation with an
AMD Ryzen 7 5800X ($8\times\SI{3.8}{\GHz}$) CPU and \SI{128}{\giga\byte} of RAM\@.
The code and data are publicly available\footnote{\url{https://github.com/tubs-alg/lawn-mowing-from-algebra-to-algorithms}}.
%Source code and data: \url{https://github.com/tubs-alg/xxx}}.
%For carrying out our experiments and evaluations, 
We used the
\emph{srpg\_iso}, \emph{srpg\_iso\_aligned} and \emph{srpg\_octa} 
instances and generated additional polyominoes with the open-source code
from the Salzburg Database of Geometric Inputs~\cite{eder2020salzburg}. 
See \cref{fig:distribution} for the overall
distribution and \cref{fig:instances} for examples.
We considered polygons with up to $n=300$ vertices and a cutter with diameter 1.
Overall, this resulted in $327$ instances. All experiments %in this section 
were carried out with a maximum runtime of \SI{300}{\second} for TSP, LNS and final IP computation.
To solve the TSP efficiently, we used the python binding \emph{pyconcorde} of the
Concorde~TSP~Solver~\cite{solverconcorde}. All components of \tspCoverage{} and \tspTurncost{} were implemented
in Python 3.10 and used the IP solver Gurobi (v10.0)~\cite{gurobi}.
As in previous work~\cite{fekete2022closer} 
the \emph{relative area} (ratio of convex hull area of $\Pol$ and cutter area $A(C)$)
is more significant for the difficulty of an instance than number of vertices of $\Pol$.
%which is why we focused
%on the variance in relative area.

\begin{figure}[h]
        \centering
        \includegraphics[width=.75\linewidth]{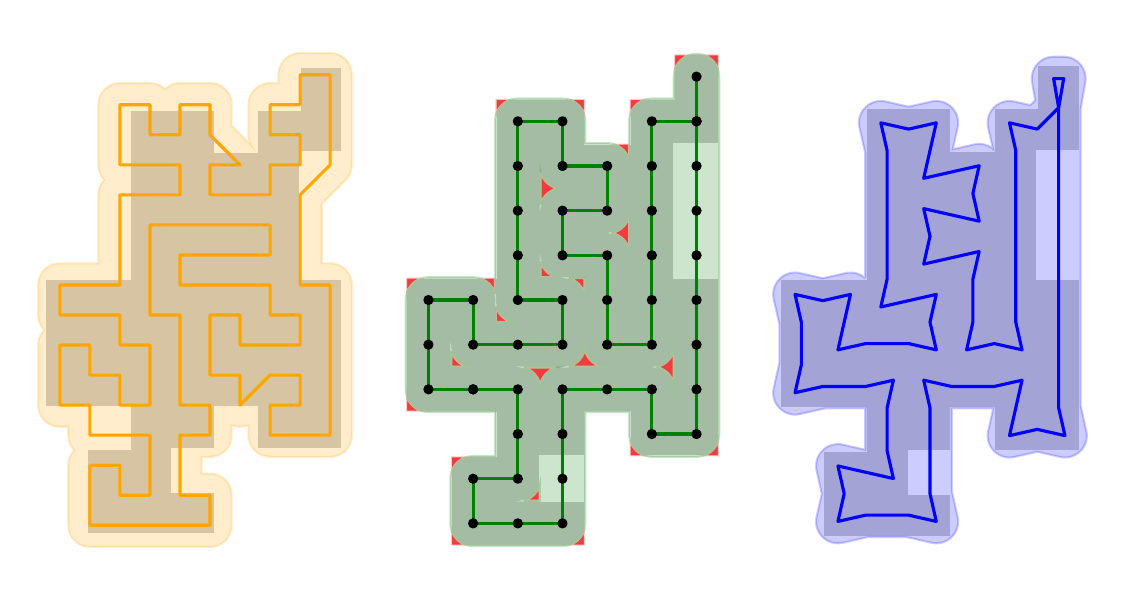}
        \caption{(Left) A \tsp{} tour yields a feasible but expensive LMP tour. (Middle)~A~TSP tour of the underlying dual grid graph, with uncovered patches shown in red. (Right)~A~feasible LMP tour after puzzle piece modification of the TSP tour.}\label{fig:eval:modification_example}
\end{figure}
\begin{figure}[t]
    \centering
      \includegraphics[width=.5\linewidth]{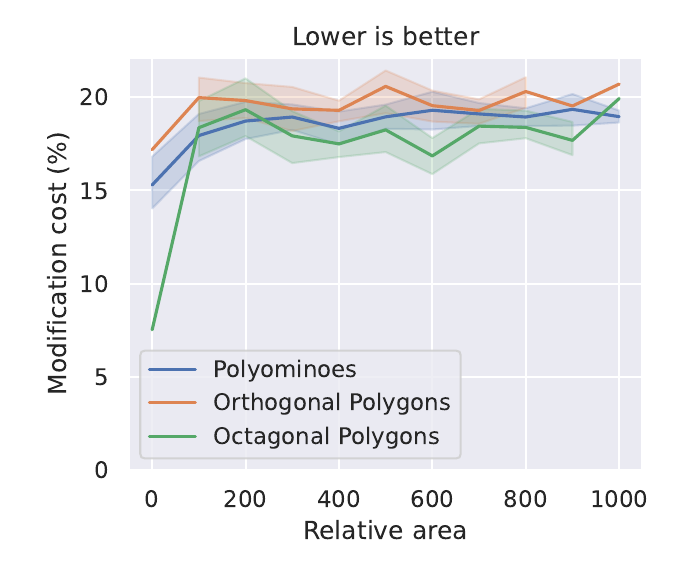}
      \caption{Modification cost over size and instance type. The modification induces a cost of around \SI{19}{\percent} over all instance types and sizes. The plot shows the average modification cost and the \SI{95}{\percent} confidence interval.}\label{fig:eval:modification_cost}
\end{figure}
\begin{figure}[h!]
    \centering
    \begin{subfigure}[b]{.5\linewidth}
      \includegraphics[width=\linewidth]{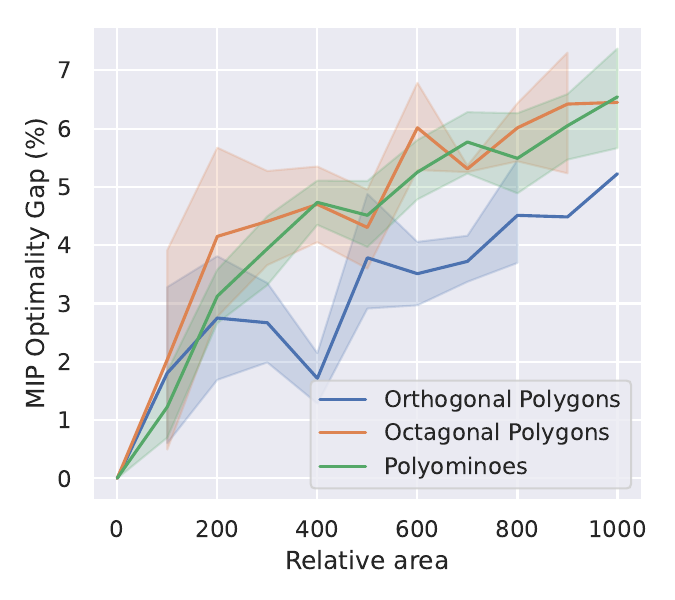}
      \caption{Optimality gap of the IP.}\label{fig:eval:optgap}
      \end{subfigure}%
    \begin{subfigure}[b]{.5\linewidth}
      \includegraphics[width=\linewidth]{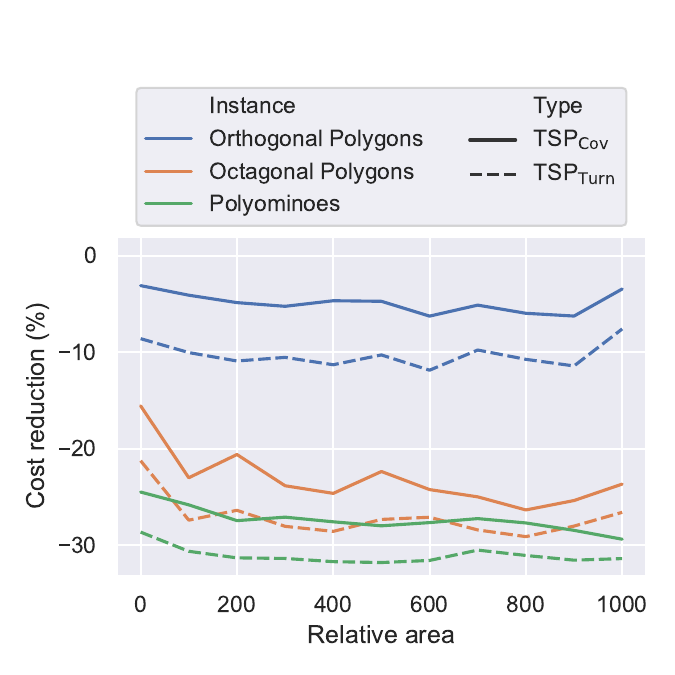}
        \caption{Cost reduction.}\label{fig:eval:cost_reduction}
      \end{subfigure} 
      \caption{
        (a) Remaining average optimality gaps for the integer program and the \SI{95}{\percent} confidence interval.
        (b) Comparison of the average cost reduction for different approaches and polygon types.
      } 
  \end{figure}

\subsection{Evaluation}
We discuss our practical results along a number of research questions (RQ).

\subparagraph{RQ1: How does \tspCoverage{} compare to \tsp{} in practice?}
%What is the average cost of the modification step?}
We compared the worst-case bound of \SI{32.6}{\percent} for \tspCoverage{} 
to the actual performance, using the total cost of \tsp{} as a baseline.
%The modification of a tour to achieve full coverage can lengthen it by up to
%\SI{61}{\percent} in the worst case, but the average case is likely to be lower.
%To evaluate this, we compare the tour in $\dualGraph$ for each instance 
%before modification to the same tour after modification.
See \cref{fig:eval:modification_example} for an example
and \cref{fig:eval:modification_cost} for the average relative modification cost.
This shows not more than an additional \SI{19}{\percent} cost, with only small 
variation over size and type. 
\Cref{fig:eval:cost_reduction} shows that the practical average reduction from \tsp{} is
independent from the size of the polygon, but differs strongly for the
different instance classes; we save $\approx\SI{27}{\percent}$ for polyominoes, $\approx \SI{24}{\percent}$ for octagonal
polygons, and $\approx \SI{5}{\percent}$ for orthogonal polygons.

%Without further optimization that incorporates the modification costs into the objective, 
%we can thus expect a relative modification cost of around \SI{19}{\percent}, which is 
%over \SI{30}{\percent} of the worst case.
\subparagraph{RQ2: How good are the solutions achieved by \tspTurncost{}?}
%The integer program allows us to compute the cheapest possible tour after modification.
%To evaluate the savings are when incorporate optimization techniques that minimize turncost, 
For the considered large instances, provably optimal solutions for the 
turn-cost minimizing IP are hard to find, so we considered %the size 
the remaining optimality gap in the IP\@.
\Cref{fig:eval:optgap} shows that gaps remain below \SI{7}{\percent} even for large instances, and 
below \SI{5}{\percent} on average for medium-sized instances.

We also compared the tours from \tspCoverage{} with the cheapest tours obtained by \tspTurncost{}
and \tsp. As shown in
%These tours can thus be considered near-optimal.
\cref{fig:eval:cost_reduction}, %we compare \tsp{}, \tspCoverage{} and \tspTurncost{} solutions.
on average we obtain  $\approx\SI{5}{\percent}$ shorter tours when compared to
the \tspCoverage{} tours, independent of instance size and type.  For 
orthogonal polygons, this doubles the cost reduction.
%Adding the modification costs to the objective, thus, can yield a visible additional improvement.

\subparagraph{RQ3: How far are we from the geometric area lower bound?}
A remaining gap between \tspTurncost{} and the area bound may
result from two sources, both from (i) the quality of the upper bound (and thus \tspTurncost{})
and (ii) the quality of the area lower bound, for the following reasons.
(i) The optimal LMP tour is not restricted to the grid graph $\dualGraph$,
so there may be cheaper tours than what we obtain from \tspTurncost{}. 
(ii) The simple area bound (corresponding to \cref{lem:area})
is relatively weak, so it is conceivable that a serious gap to this lower bound remains.

%To evaluate how far we are at most from this tour, we compare our tours with a geometric 
%%lower bound as we cannot compute an optimal tour in reasonable time.
%We obtain a lower bound by the simple observation that moving a circular tool with 
%diameter $d$ by $\epsilon$ can cover at most an area of $\epsilon\cdot d$.
%A coverage tour of a polygon with area $A$ and unit diameter, thus, has to have a length of at 
%least $A-\frac{\pi}{4}$.
Overall, the combination of both effects remains limited,
as can be seen from \Cref{fig:eval:geomlb} (showing the ratio of \tspTurncost{} value and area bound):
For the octagonal polygons and polyominoes, we are on average at most \SI{50}{\percent} above the area bound.
%the geometric shortest possible tour with our integer program.
For orthogonal polygons, the relative gap %length of \tspTurncost{} tours 
is on average below \SI{80}{\percent}.

\begin{figure}[t]
  \centering
  \begin{subfigure}[b]{.5\linewidth}
    \centering
    \includegraphics[width=\linewidth]{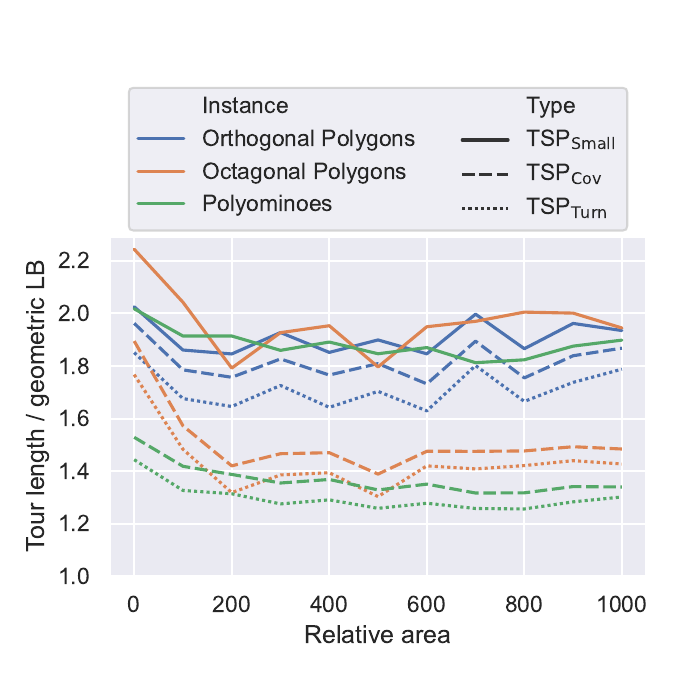}
    \caption{Tour length compared to the area bound.}
    \label{fig:eval:geomlb-a}
  \end{subfigure}%
  \begin{subfigure}[b]{.5\linewidth}
    \centering
    \includegraphics[width=\linewidth]{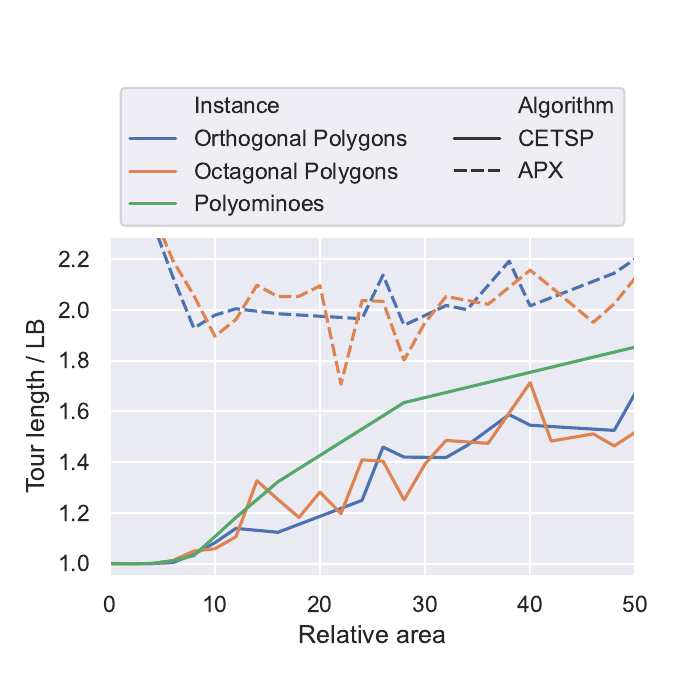}
    \caption{Solution quality of~\cite{fekete2022closer}.}
    \label{fig:eval:geomlb-b}
\end{subfigure}
  \caption{
  (a) Tour length compared to the (weaker) area lower bound in terms of the average ratio.
  For octagonal polygons and polyominoes, we can get below \SI{50}{\percent} on average.
  (b)~Comparable results for the average solution quality of~\cite{fekete2022closer} 
based on a (stronger) CETSP bound; here APX denotes the performance of
the approximation
algorithm by~\cite{Arkin2000}. Note the considerably larger relative area in
comparison to~\cite{fekete2022closer}.
}\label{fig:eval:geomlb}
\end{figure}

\subparagraph{RQ4: How do our solutions compare to previous practical work?}

As shown before, our results are already considerably better than work 
based on \tsp. A comparison to the previous best practical results by Fekete et al.~\cite{fekete2022closer}
(whose instances were used as a subset of our benchmarks)
is shown in \cref{fig:eval:geomlb}; plotted are the ratios between the achieved solution values
and the respective lower bounds. Fekete~et~al.~\cite{fekete2022closer} employ a more sophisticated 
lower bound based on an evaluation of
a series of Close-Enough TSP (CETSP) instances.
The authors pointed out that the lower bound computation becomes very expensive even for instances with relative area smaller than $50$, see Figures~\num{11}~and~\num{12}~in~\cite{fekete2022closer}. 
Because we evaluate much larger instances, our ratios only use the relatively straightforward area bound.
As a consequence, the denominators of these ratios favor the evaluation
for~\cite{fekete2022closer}, which are shown in~\cref{fig:eval:geomlb-b};
see \cref{fig:alenex-a,fig:alenex-b} for a comparison on a relatively small example
that was also shown in \cite{fekete2022closer}.
In addition, we were able to achieve results for instances with a relative area
20 times larger than~\cite{fekete2022closer}.

Despite these additional challenges (of weaker bounds and larger instance sizes), our results 
compare favorably to the ones reported by~\cite{fekete2022closer}.
The main reason lies in our structurally simpler approach that
still yields good results when the complex evaluation of the CETSP 
from~\cite{fekete2022closer} reaches computational limitations. 
As can be seen from a comparison of computed trajectories for the
visual example (\cref{fig:alenex-c,fig:alenex-d}), this 
is also reflected in simpler trajectories obtained from \tspTurncost{}.

\begin{figure}[h!]
  \centering
  \begin{subfigure}[b]{.25\linewidth}
    \centering
    \includegraphics[width=\linewidth]{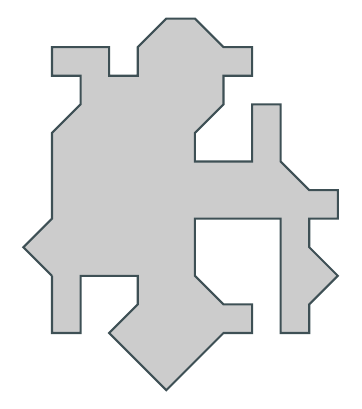}
    \caption{Area (LB): $36.25$}
  \label{fig:alenex-a}
  \end{subfigure}%
  \begin{subfigure}[b]{.25\linewidth}
    \centering
    \includegraphics[width=\linewidth]{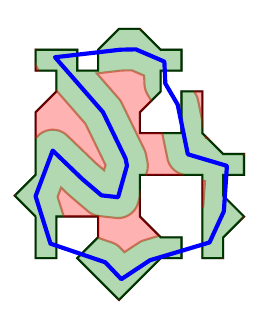}
    \caption{LB\@: $40.94$ \cite{fekete2022closer}}
  \label{fig:alenex-b}
  \end{subfigure}%
  \begin{subfigure}[b]{.25\linewidth}
    \centering
    \includegraphics[width=\linewidth]{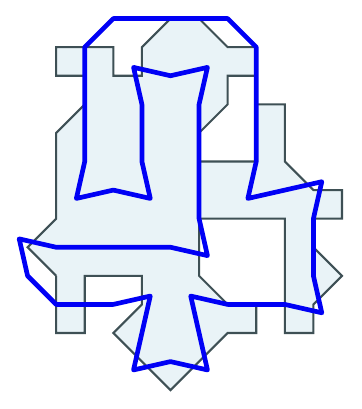}
    \caption{\tspTurncost{} (UB): $66.71$}
  \label{fig:alenex-c}
  \end{subfigure}%
  \begin{subfigure}[b]{.25\linewidth}
    \includegraphics[width=\linewidth]{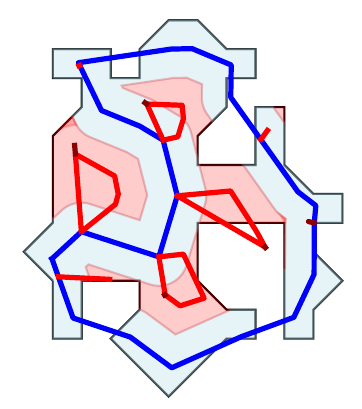}
    \caption{UB\@: $68.16$ \cite{fekete2022closer}}
  \label{fig:alenex-d}
  \end{subfigure}
  \caption{Comparison of \tspTurncost{} with lower and upper bounds from 
  Fekete~et~al.~\cite{fekete2022closer}.}
  \label{fig:alenex}
\end{figure}

\section{Conclusion}
\label{sec:conclusion}

We have presented new insights for the Lawn Mowing Problem,
starting with an algebraic analysis of the structure of
optimal trajectories. As a consequence, we can pinpoint
a particular source of the perceived overall difficulty
of the problem, and prove that constructing optimal
tours necessarily involves operations that go beyond
simple geometric means; we can also use these insights
to come up with better construction methods for tours,
both on the theoretical and the practical side,
with minimizing overall turn cost playing a crucial role.

Our results also clear the way for a number of important
followup questions. Is it possible to improve our
approach for polyominoes? As discussed in the text,
considering higher-order connectivity between turns
and using slightly off-center, axis-parallel strips
appears to be a relatively easy way for (albeit marginal)
improvement. It may very
well be that this ultimately leads to optimal
tours for polyominoes; however, final success
on this fundamental challenge will require
another breakthrough in establishing lower bounds,
as neither the polygon area (which may incur 
a gap from the optimal value, similar to the number
of vertices in a grid graph does from a TSP solution)
nor the Close-Enough TSP bound for a finite set
of witness points may suffice to certify optimality.
Given that an optimal tour may also involve portions
that are not axis-parallel, it will also require
further algebraic analysis of turns that are not multiples
of 90 degrees.

For the Lawn Mowing Problem on general regions
(which may not even have to be connected), our 
hardness result hints at further difficulties.
It is quite conceivable that the general LMP
is not just algebraically hard, but even $\exists\R$-complete.
Even in that case, we believe that further engineering
of the tile-based mowing of polyominoes (with attention
to turn cost) and Close-Enough TSP
may be the most helpful tools for further systematic improvement.

\bibliography{bibliography}

\clearpage
\appendix

\section{Source Code for Algebraic Verification of \cref{lemma:optimal-path-square-with-circle}}\label{sec:additional-proof-content-optimal-path}

\begin{lstlisting}[language=Mathematica,caption={Mathematica source code for the calculating the optimal values for $\fociPoint,\upperLeftPoint,\leavePoint$}]
r = 1 / 2;
xPos[delta_] := 3/2* r
yPos[delta_] := 1/2 * (r + delta*r + delta*r)

Dist[x_, y_, z_, v_] := Sqrt[(x - z)^2 + (y - v)^2]
Angle[x_, y_, z_, v_] := VectorAngle[{x - z, y - v}, {1, 0}]

Theta[x_, y_, delta_] := Pi - Angle[r, delta*r, 2 r, r + delta*r]
f[delta_] := Dist[xPos[delta], yPos[delta], r, r*delta] 
a[x_, y_, delta_] :=  1/2  * (Dist[x, y, r, delta*r] +  Dist[x, y, 2 r, r + delta*r])
b[x_, y_, delta_] := Sqrt[a[x, y, delta]^2 - f[delta]^2] 

ellipseEquation[x_, y_, delta_] := ((x - xPos[delta])* Cos[Theta[x, y, delta]] + (y - yPos[delta])* Sin[Theta[x, y, delta]])^2/ a[x, y, delta]^2 + ((x - xPos[delta])* Sin[Theta[x, y, delta]] - (y - yPos[delta])* Cos[Theta[x, y, delta]])^2/b[x, y, delta]^2

distanceToQ[x_, y_, delta_] := Dist[x, y, r, delta*r] +  Dist[x, y, 2 r, r + delta*r]
distanceToCircleCenter[x_, y_] :=  x^2 + (y - 2 *r)^2

result = 
  Minimize[{c + delta*r, 
    distanceToCircleCenter[x, y] == r^2 &&  
     ellipseEquation[x, y, delta] == 1 && 
     distanceToQ[x, y, delta] - c == 0 && 0 <= delta <= 1 && 
     0 <= x <= r && r <= y <= 2 r && c > 0}, {x, y, c, delta}];

yReduce = RootReduce[y /. Last@result];
xReduce = RootReduce[x /. Last@result];
deltaReduce = RootReduce[delta /. Last@result];
\end{lstlisting}

\clearpage
\section{Additional Figures}

\begin{figure}[h]
  \centering
  \includegraphics[height=.5\paperheight, trim={1cm 2.5cm 1cm 2.5cm}, clip]{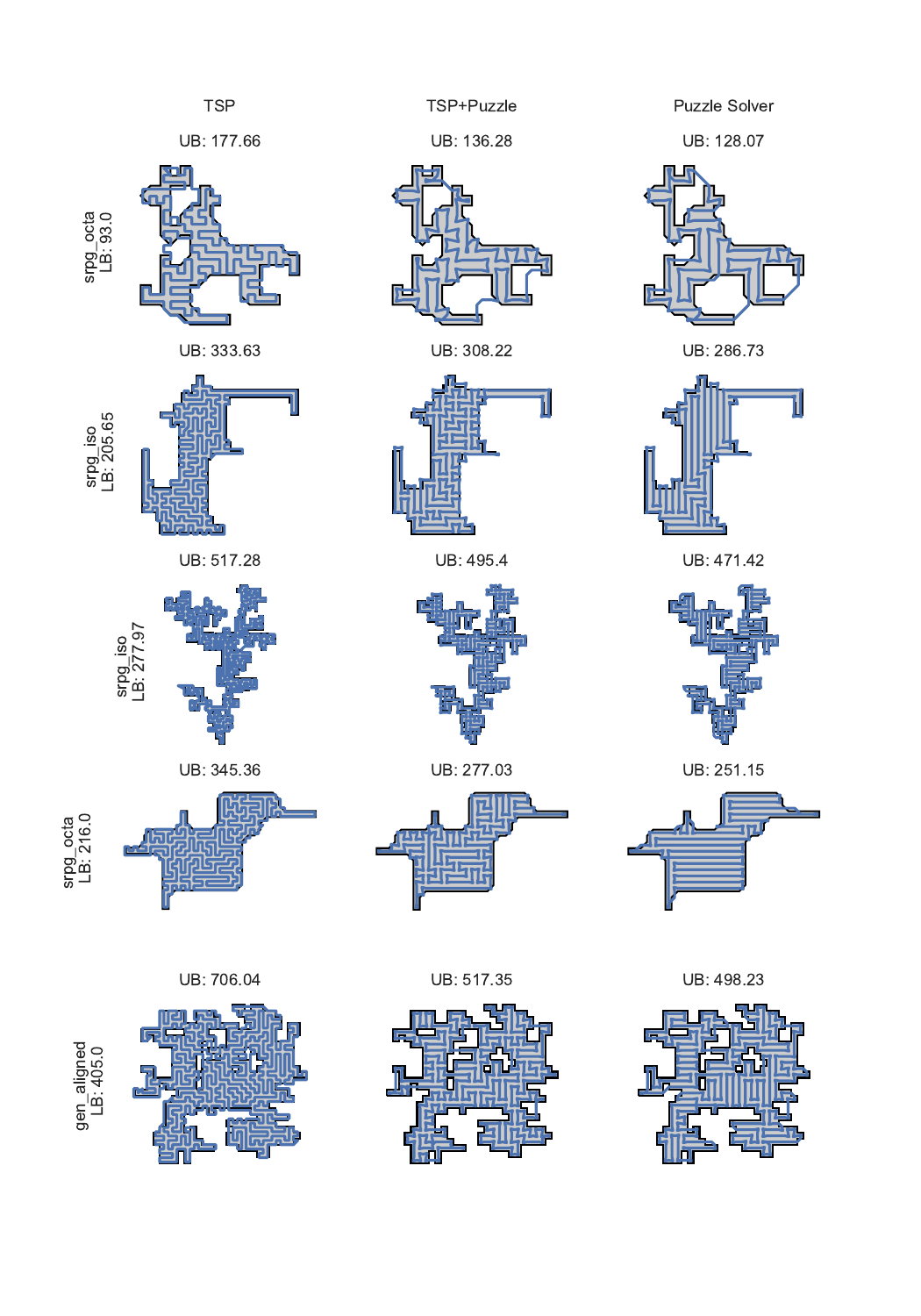}
  \caption{Examples of different solutions. Left column is the \tsp{} solver which operates on a $\sqrt{2}/2$ grid. The middle column shows
  solutions to \tspCoverage{} which operates on a larger grid. 
  The right column shows \tspTurncost{} solutions
  which modify and improve the \tspCoverage{} solution by minimizing the number of turns.}
  \label{fig:solutions}
\end{figure}

\begin{figure}[h]
  \centering
  \includegraphics[height=.6\paperheight, trim={0 2.5cm 0 0}, clip]{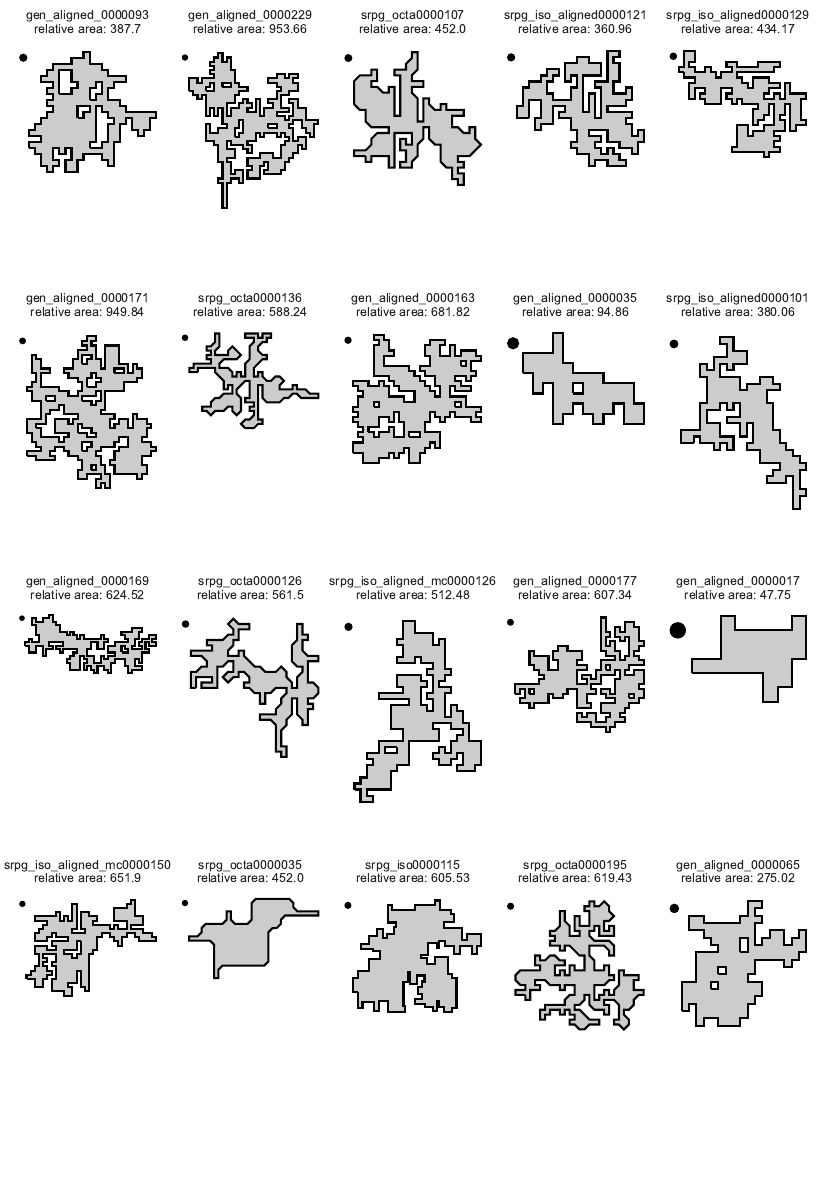}
  \caption{Examples of the used polygons. The black circles show the respective cutter size; the \emph{relative area} denotes
  the ratio of convex hull area and cutter area.}
  \label{fig:instances}
\end{figure}

\end{document}